\documentclass[final,12pt]{colt2025} % 

\usepackage{amsfonts,mathtools,mathrsfs,color}
\usepackage{caption}
\usepackage{subcaption}

\makeatletter
\newtheorem*{rep@theorem}{\rep@title}
\newcommand{\newreptheorem}[2]{%
\newenvironment{rep#1}[1]{%
 \def\rep@title{#2 \ref{##1}}%
 \begin{rep@theorem}}%
 {\end{rep@theorem}}}
\makeatother

\newtheorem*{theorem*}{Theorem}

\newreptheorem{theorem}{Theorem}
\newreptheorem{lemma}{Lemma}

\renewcommand{\P}{\mathcal{P}}
\newcommand{\R}{\mathbb{R}} 
\newcommand{\sR}{\mathsf{R}} 
\newcommand{\F}{\mathcal{F}}
\newcommand{\N}{\mathcal{N}}

\newcommand{\E}{\mathbb{E}}

\renewcommand{\part}[2]{\frac{\partial #1}{\partial #2}}

\newcommand{\HS}{\mathrm{HS}}

\newcommand{\Tr}{\mathsf{Tr}}

\newcommand{\Var}{\mathrm{Var}}

\newcommand{\KL}{\mathsf{KL}}

\newcommand{\FI}{\mathsf{FI}}
\newcommand{\grad}{\mathsf{grad}}
\newcommand{\Hess}{\mathsf{Hess}}

\newcommand{\sym}{\mathrm{sym}}
\newcommand{\ess}{\mathrm{ess}}

\newcommand{\red}[1]{{\color{red}{#1}}}
\newcommand{\blue}[1]{{\color{blue}{#1}}}
\newcommand{\purple}[1]{{\color{purple}{#1}}}

\title[Mixing Time of the Proximal Sampler in Relative Fisher Information via SDPI]
{Mixing Time of the Proximal Sampler in Relative Fisher Information \\
via Strong Data Processing Inequality}
\usepackage{times}

\coltauthor{%
 \Name{Andre Wibisono} \Email{andre.wibisono@yale.edu}\\
 \addr Department of Computer Science,
 Yale University, New Haven, CT, USA
}

\begin{document}

\maketitle

\begin{abstract}%
    We study the mixing time guarantee for sampling in relative Fisher information via the Proximal Sampler algorithm, which is an approximate proximal discretization of the Langevin dynamics. We show that when the target probability distribution is strongly log-concave, the relative Fisher information converges exponentially fast along the Proximal Sampler; this matches the exponential convergence rate of the relative Fisher information along the continuous-time Langevin dynamics for strongly log-concave target.
    When combined with a standard implementation of the Proximal Sampler via rejection sampling, this exponential convergence rate provides a high-accuracy iteration complexity guarantee for the Proximal Sampler in relative Fisher information when the target distribution is strongly log-concave and log-smooth.
    Our proof proceeds by establishing a strong data processing inequality for relative Fisher information along the Gaussian channel under strong log-concavity, and a data processing inequality along the reverse Gaussian channel for a special distribution.
    The forward and reverse Gaussian channels compose to form the Proximal Sampler, and these data processing inequalities imply the exponential convergence rate of the relative Fisher information along the Proximal Sampler.\footnote{Extended abstract accepted for presentation at Conference on Learning Theory (COLT) 2025}
\end{abstract}


%%%%%%%%
\setcounter{tocdepth}{2}
\tableofcontents

%%%%%%%%%
\section{Introduction}
\label{Sec:Introduction}

In this paper we establish the mixing time guarantee of the Proximal Sampler algorithm for sampling from a strongly log-concave target distribution with error guarantee in relative Fisher information.
We describe the background from sampling and the analogy with optimization in Section~\ref{Sec:IntroSampling}.
Our technical contribution is to establish a strong data processing inequality in relative Fisher information under strong log-concavity, which we describe in Section~\ref{Sec:IntroSDPI}.

\subsection{Mixing Time of the Proximal Sampler in Relative Fisher Information}
\label{Sec:IntroSampling}

Sampling from a probability distribution is a fundamental algorithmic task in many applications, including in computer science, machine learning, and Bayesian statistics.
Sampling algorithms typically take the form of random walks or Markov chains which have the target probability distribution as their stationary distribution.
The \textit{mixing time} of the Markov chain, i.e., how fast the Markov chain converges to its stationary distribution, determines the iteration complexity of the sampling algorithm.
For many applications, we need mixing time guarantees where the error is measured in a certain metric or statistical divergence, such as in the total variation (TV) distance, Wasserstein distance, Kullback-Leibler (KL) divergence, chi-square ($\chi^2$) divergence, R\'enyi divergence, or even in the general $\Phi$-divergence~\citep{chafai2004entropies}, which broadly generalizes the TV distance, KL divergence, and $\chi^2$-divergence.
Thus, there have been great interests in developing Markov chains with fast mixing time guarantees in various divergences which can be implemented efficiently in practice.

For target distributions supported on a continuous domain such as the Euclidean space, many sampling algorithms can be derived as discretizations of continuous-time stochastic processes such as the  Langevin dynamics; these include the Unadjusted Langevin Algorithm (ULA)~\citep{roberts1996exponential} and the Proximal Sampler algorithm~\citep{LST21}, which we also study in this paper.
The Langevin dynamics for sampling has a natural optimization interpretation as the gradient flow for minimizing KL divergence in the space of probability distributions~\citep{JKO98}.
This perspective of sampling as optimization is useful to translate the tools and techniques from optimization to sampling, to derive and analyze sampling algorithms in discrete time, and to suggest sampling guarantees one may hope to obtain by drawing parallels from optimization; see~\citep{W18}, and Section~\ref{Sec:SamplingAsOpt} for a review.

When the target distribution is well-behaved, e.g., when it is strongly log-concave 
or satisfies isoperimetric conditions such as the log-Sobolev inequality (LSI), existing results in the literature have established mixing time guarantees in various divergences.
In particular, for strongly log-concave target distributions, the Langevin dynamics in continuous time is known to converge exponentially fast in Wasserstein distance and KL divergence~\citep{bakry1985diffusions}, in R\'enyi divergence~\citep{cao2019exponential}, and in all $\Phi$-divergence~\citep{chafai2004entropies}.
In discrete time, the convergence guarantees are more subtle due to a potential bias issue, since some algorithms may have stationary distributions which are different from the target distribution.

A simple discretization of the Langevin dynamics is the \textit{Unadjusted Langevin Algorithm (ULA)}, which is an explicit algorithm that is easy to implement in practice.
However, ULA is \textit{biased}, which means for each fixed step size, it converges to a biased limiting distribution~\citep{roberts1996exponential}; this bias can be seen as arising from using a mismatched discretization method for discretizing the Langevin dynamics when viewed as a gradient flow for optimizing a composite objective function~\citep{W18}.
When the target distribution is strongly log-concave and log-smooth, one can show the biased convergence guarantees of ULA to the target distribution in Wasserstein distance~\citep{dalalyan17further,dalalyan2017theoretical,durmus2017nonasymptotic}, in KL divergence~\citep{cheng2018convergence,durmus2019analysis}, or in R\'enyi divergence~\citep{chewi2024analysis};
one can also show the exponential convergence guarantees of ULA to its biased stationary distribution in KL and R\'enyi divergence~\citep{VW19, AltschulerTalwar23}, and in $\Phi$-divergence~\citep{mitra2025fast}.
However, the presence of the bias only yields a ``low-accuracy'' iteration complexity guarantee for ULA, where the number of iterations to reach an error $\varepsilon$ in the specified divergence scales polynomially in $\varepsilon^{-1}$, rather than logarithmically as suggested by the exponential convergence in continuous time.\footnote{
One way to remove the bias of ULA is to add a Metropolis-Hastings filter or the accept-reject step, resulting in the Metropolis-Adjusted Langevin Algorithm (MALA)~\citep{roberts1996exponential}.
Since MALA is unbiased, its iteration complexity has the desired logarithmic dependency on the error parameter; however, the convergence guarantees of MALA are only known in weaker metrics such as in the total variation distance~\citep{dwivedi2019log,chen2020fast,chewi2021optimal,wu2022minimax,chen2023simple}.
}

Another algorithm that implements the Langevin dynamics is the \textit{Proximal Sampler}; it was derived as a Gibbs sampling algorithm for sampling from a joint distribution defined on an extended phase space~\citep{titsias2018auxiliary,LST21}, and it can be viewed as an approximate proximal discretization of the Langevin dynamics~\citep{LST21,CCSW22}.
The Proximal Sampler is unbiased, so it converges to the correct target distribution.
When the target distribution is strongly log-concave, one can show an exponential convergence guarantee of the Proximal Sampler in Wasserstein distance~\citep{LST21}, in KL and R\'enyi divergence~\citep{CCSW22}, and in all $\Phi$-divergence~\citep{mitra2025fast}.
The Proximal Sampler is an implicit algorithm, since in each iteration it requires sampling from a regularized distribution, known as the \textit{Restricted Gaussian Oracle (RGO)}~\citep{LST21}; when the target distribution is log-smooth, the RGO can be implemented efficiently via rejection sampling~\citep{CCSW22}. 
Thus, for strongly log-concave and log-smooth target distributions, the Proximal Sampler has a ``high-accuracy'' iteration complexity guarantee that scales logarithmically with the inverse error parameter, as expected from the exponential convergence of the continuous-time Langevin dynamics.

% \noindent
% \textbf{Our work.~}
\paragraph{Our work.}
In this paper, we study mixing time guarantees for sampling in \textit{relative Fisher information} via the Proximal Sampler.
The relative Fisher information between two distributions measures the expected square distance between their score functions (gradient of the log-density);
in contrast, KL divergence or $\Phi$-divergence only compares their log-density or the ratio of their density functions.
Relative Fisher information is not a $\Phi$-divergence, so previous results for $\Phi$-divergence do not apply.
Relative Fisher information is often larger than other divergences:
when the target distribution satisfies LSI or Poincar\'e inequality, relative Fisher information is larger than KL divergence or the squared TV distance, respectively.
Even when the target distribution is strongly log-concave, the gap between KL divergence and relative Fisher information can be arbitrarily large; see Lemma~\ref{Lem:LargeFI} in Section~\ref{Sec:LargeFIEx} for an example.
Thus, mixing time guarantees in relative Fisher information are valuable because they are generally stronger than guarantees in KL divergence or TV distance.

\paragraph{Analogy with optimization.}
Another motivation for our work is to strengthen the links between sampling and optimization,
in particular via the perspective of sampling as optimization.
We can view sampling from a target distribution $\nu$ over $\R^d$ 
as the problem of optimizing the KL divergence $\KL(\cdot \,\|\, \nu)$ on the space of probability distributions $\P(\R^d)$,
and the gradient flow under the Wasserstein metric for optimizing KL divergence is the Langevin dynamics~\citep{JKO98}.
Relative Fisher information has the interpretation as the squared Wasserstein gradient norm of the KL divergence.
In the Wasserstein geometry, conditions on the distribution $\nu$ translate to properties of the KL divergence that enable fast optimization; e.g., the strong log-concavity of $\nu$ corresponds to the geodesic strong convexity of the KL divergence, while the LSI on $\nu$ corresponds to the Polyak-\L{}ojaciewicz or gradient domination condition of KL divergence~\citep{OV00}.
See Section~\ref{Sec:SamplingAsOpt} for further review.

In strongly convex optimization, we can obtain exponential convergence guarantees in both the function value and in the squared gradient norm, in continuous time via the gradient flow dynamics and in discrete time via the proximal gradient algorithm (see Section~\ref{Sec:ReviewOpt} for a review).
Then for sampling a strongly log-concave target distribution, 
it is natural to expect exponential convergence guarantees in both the KL divergence (the objective function), as well as in the relative Fisher information (the squared gradient norm).
The exponential convergence guarantees in KL divergence are known both for the Langevin dynamics in continuous time (see Section~\ref{Sec:LangevinKLRate}) and for the Proximal Sampler in discrete time (see Section~\ref{Sec:ProximalSamplerReviewKL}).
When the target distribution is strongly log-concave, the Langevin dynamics is also known to converge exponentially fast in relative Fisher information (see Section~\ref{Sec:LangevinReview}).
However, in discrete time, the corresponding exponential convergence guarantee in relative Fisher information seems to be unknown; this is what we fill out in this work.
Indeed, we show that 
the Proximal Sampler has an exponential convergence guarantee in relative Fisher information when the target distribution is strongly log-concave (Theorem~\ref{Thm:FIProximalRate}).
% in Section~\ref{Sec:ProximalSampler}).
This implies a high-accuracy iteration complexity guarantee of the Proximal Sampler in relative Fisher information for log-smooth and strongly log-concave target distribution (Corollary~\ref{Cor:IterationComplexity}).

\paragraph{Related work.}
A few works have studied mixing time in relative Fisher information as a sampling analogue of the gradient norm convergence in optimization.
\cite{balasubramanian22a} study mixing time in relative Fisher information for sampling from non-log-concave and log-smooth target distributions, analogous to finding approximate stationary points in non-convex optimization;
they study the biased convergence of ULA in relative Fisher information, and provide a low-accuracy iteration complexity guarantee on the average iterate of ULA.
\cite{chewi2023fisher} study the iteration complexity lower bounds for relative Fisher information in non-log-concave sampling, analogous to lower bounds for finding stationary points in non-convex optimization.
We note that~\cite[Appendix~A.2]{chewi2023fisher} provide an argument for obtaining a high-accuracy mixing time in relative Fisher information for log-concave target distributions, involving a reduction from the high-accuracy mixing time of the Proximal Sampler in $\chi^2$-square divergence from~\citep{CCSW22} and a post-processing via the heat flow.
In this work, we give a direct proof of the exponential convergence of the Proximal Sampler in relative Fisher information for strongly log-concave target distributions.
We provide further discussion on related works in Section~\ref{Sec:RelatedWorks}.

\subsection{Strong Data Processing Inequality in Relative Fisher Information}
\label{Sec:IntroSDPI}

\textit{Data processing inequality} (DPI) is
a fundamental tenet in information theory which states that information is never increasing along any noisy channel~\citep{cover2006elements}.
Here ``information'' is measured via statistical divergences such as the KL divergence or any $\Phi$-divergence, and their monotone transformations including R\'enyi divergence.
% 
% It is unclear whether DPI holds for relative Fisher information, which is not a $\Phi$-divergence.\footnote{
% We note we work with the non-parametric relative Fisher information, which has a natural geometric interpretation as the squared gradient norm of KL divergence in Wasserstein metric; for the classical parametric Fisher information, some form of DPI is known, see e.g.~\citep{zamir1998proof} and Section~\ref{Sec:RelatedWorksFI} for further discussion.
% }
% In fact, we show that DPI in relative Fisher information does not always hold in general, by constructing a counterexample for the Gaussian channel (see Example~\ref{Ex:CounterDecrHeat}).
% 
The {\em strong data processing inequality} (SDPI) is a strengthening of DPI that quantifies the rate at which information is contracting along the noisy channel, where the contraction rate 
may depend on properties of the channel and the input distributions.
For a review of SDPIs and their applications, see e.g.~\citep{raginsky2016strong,polyanskiy2017strong,caputo2024entropy,PW_book}.

SDPIs are useful because many sampling algorithms can be viewed as composition of noisy channels.
In continuous time, the Langevin dynamics can be viewed as a channel driven by a diffusion process.
In discrete time, ULA can be viewed as the composition of a deterministic gradient map and a Gaussian channel, while the Proximal Sampler can be viewed as the composition of the forward and reverse Gaussian channels.
If we have an SDPI for each channel in the algorithm, then we immediately obtain a mixing time guarantee for the sampling algorithm by iterating the contraction rate across multiple steps of the algorithm.
This strategy has been successful in establishing the convergence guarantees of ULA and the Proximal Sampler in KL and R\'enyi divergence, e.g., by~\citep{VW19,CCSW22,YLFWC23,kook2024inandout}, in $\Phi$-divergence by~\citep{mitra2025fast}, and in mutual information by~\citep{LMW24}.

For our goal of proving mixing time guarantees in relative Fisher information, ideally we want to establish SDPIs in relative Fisher information.
Relative Fisher information is not a $\Phi$-divergence, nor a monotone tranformation of it.
It seems unknown whether DPI holds for relative Fisher information.\footnote{
We work with the non-parametric relative Fisher information, which has a natural geometric interpretation as the squared gradient norm of KL divergence in Wasserstein metric; for the classical parametric Fisher information, some form of DPI is known, see e.g.~\citep{zamir1998proof} and also Section~\ref{Sec:RelatedWorksFI} for further discussion. 
We also note that one way to prove DPI for $\Phi$-divergence is to use the joint convexity of the $\Phi$-divergence and Jensen's inequality;
this fails since relative Fisher information is convex in its first argument, but not convex in its second argument~\citep{W11}.}
In fact, we show that DPI in relative Fisher information does not always hold in general, by constructing a counterexample for the Gaussian channel (see Example~\ref{Ex:CounterDecrHeat}).

To proceed, we focus on the \textit{Fokker-Planck (FP) channels}, a family of noisy channels generated by continuous-time diffusion processes (see Section~\ref{Sec:FPChannel}); this includes the Langevin dynamics, the Gaussian channel, and the reverse Gaussian channel, which compose the Proximal Sampler.
One way to prove SDPIs for $\Phi$-divergence along the FP channels is via time-differentiation of the $\Phi$-divergence along the simultaneous evolution along the FP channel
(see Lemma~\ref{Lem:TimeDerivativeKL} for a review in the case of KL divergence).
We extend this technique to analyze the time derivative of the relative Fisher information along a simultaneous evolution along the FP channel.
Under log-concavity assumptions on one of the input distributions, we can extract DPI and SDPI results for relative Fisher information.

\subsection{Contributions}

Our main algorithmic result is the exponential convergence of the Proximal Sampler in relative Fisher information when the target distribution is strongly log-concave (Theorem~\ref{Thm:FIProximalRate}).
We combine this with a standard rejection sampling implementation of the Proximal Sampler to conclude a high-accuracy iteration complexity guarantee of the Proximal Sampler in relative Fisher information for log-smooth and strongly log-concave target distribution (Corollary~\ref{Cor:IterationComplexity}).

Our main technical contribution is to establish a formula for the time derivative of relative Fisher information along a simultaneous application of the Fokker-Planck channel (Lemma~\ref{Lem:TimeDerivativeFI}).
We use this to show that along the Gaussian channel, (S)DPI in relative Fisher information holds when the second distribution is (strongly) log-concave.
We also show that along the Gaussian channel,
SDPI in relative Fisher information eventually holds when the second distribution is a log-Lipschitz perturbation of a strongly log-concave distribution (Theorem~\ref{Thm:FisherInfoHeatFlow}).
However, in general, DPI in relative Fisher information may not hold. 
We show an example when DPI initially
does not hold along the Gaussian channel even when the first distribution is a Gaussian (Example~\ref{Ex:CounterDecrHeat}).
% We also show a restricted SDPI in relative Fisher information holds along the reverse Gaussian channel, and use this fact to prove Theorem~\ref{Thm:FisherInfoHeatFlow}.

We also show that along the Ornstein-Uhlenbeck (OU) channel, SDPI in relative Fisher information eventually holds when the second distribution is strongly log-concave (Theorem~\ref{Thm:FisherInfoOUFlow}).
However, we also show an example when DPI in relative Fisher information initially does not hold along the OU channel, even when both input distributions are Gaussians (Example~\ref{Ex:SDPIOU}).

%%%%%%%%%%%%%%%%%
\section{Preliminaries and Reviews}
\label{Sec:Preliminary}

%%%%%%%%%%%%%%%%%%%%
\subsection{Notations}
\label{Sec:Notations}

In this paper, we are working on the Euclidean space $\R^d$ for some $d \ge 1$.
We denote the $\ell_2$-inner product between vectors $u = (u_1,\dots,u_d),v = (v_1,\dots,v_d) \in \R^d$ as $u^\top v = \langle u,v \rangle = \sum_{i=1}^d u_i v_i$.
For a symmetric matrix $A \in \R^{d \times d}$, the notation $A \succeq 0$ means $A$ is positive semidefinite, i.e., $u^\top A u \ge 0$ for all $u \in \R^d$.
For symmetric matrices $A, B \in \R^{d \times d}$, the notation $A \succeq B$ means $A-B \succeq 0$ is positive semidefinite.
Throughout, let $I \in \R^{d \times d}$ denote the identity matrix.

For a matrix $A \in \R^{d \times d}$, let $\Tr(A) = \sum_{i=1}^d A_{ii}$ be the trace of $A$.
Let $\|A\|_{\HS}^2 = \Tr(A^\top A)$ denote the squared Hilbert-Schmidt (Frobenius) norm of $A$.
For $u \in \R^d$ and $A \in \R^{d \times d}$, we denote 
$$\|u\|^2_A := u^\top A u$$ 
to be the ``squared norm'' of a vector $u \in \R^d$ with respect to the matrix $A$.
If $A \succeq 0$, then $\|u\|^2_A \ge 0$ for all $u \in \R^d$.
However, if $A$ is not positive semidefinite, then $\|u\|^2_A$ may be negative for some $u$.
Note $\|u\|^2_A$ only depends on the symmetric part $A_{\textsf{sym}} = \frac{1}{2}(A+A^\top)$ of $A$, i.e., $\|u\|^2_A = \|u\|^2_{A_\sym}$.

For a differentiable function $f \colon \R^d \to \R$, let $\nabla f(x) \in \R^d$ denote the gradient vector at $x \in \R^d$ of the partial derivatives: $(\nabla f(x))_i = \part{f(x)}{x_i}$.
Let $\nabla^2 f(x) \in \R^{d \times d}$ be the Hessian matrix of second partial derivatives: $(\nabla^2 f(x))_{i,j} = \part{^2 f(x)}{x_i \, x_j}$.
Let $\Delta f(x) = \Tr(\nabla^2 f(x))$ be the Laplacian.

For a vector field $v \colon \R^d \to \R^d$ with $v(x) = (v_1(x),\dots,v_d(x)) \in \R^d$,
let $\nabla v \colon \R^d \to \R^{d \times d}$ be the Jacobian matrix of mixed partial derivatives: $(\nabla v(x))_{i,j} = \part{v_i(x)}{x_j}$.
Let $\nabla \cdot v \colon \R^d \to \R$ be the divergence of $v$, defined by
$(\nabla \cdot v)(x) = \sum_{i=1}^d \part{v_i(x)}{x_i} = \Tr(\nabla v(x))$.

%%%%%%%%%%%%%
\subsection{Relative Entropy and Relative Fisher information}

Let $\P(\R^d)$ denote the space of probability distributions $\rho$ over $\R^d$ which are absolutely continuous with respect to the Lebesgue measure and have a finite second moment $\E_{\rho}[\|X\|^2] < \infty$.
We identify a probability distribution $\rho \in \P(\R^d)$ with its probability density function with respect to the Lebesgue measure, which we also denote by $\rho \colon \R^d \to \R$, so $\rho(x) > 0$ and $\int_{\R^d} \rho(x) dx = 1$.
We say $\rho$ is absolutely continuous with respect to another distribution $\nu$, denoted by $\rho \ll \nu$, if $\nu(A) = 0$ implies $\rho(A) = 0$ for any $A \subseteq \R^d$; if $\rho$ and $\nu$ both have density functions, then $\rho \ll \nu$ means $\nu(x) = 0$ implies $\rho(x) = 0$ for all $x \in \R^d$.
Throughout, we assume $\rho$ and $\nu$ have full support $\R^d$ and their density functions have the required regularity properties for the various operations below to be well-defined.

Let $H \colon \P(\R^d) \to \R$ be the (differential) \textit{entropy} functional:
$$H(\rho) = -\E_\rho[\log \rho] = -\int_{\R^d} \rho(x) \, \log \rho(x) \, dx.$$

Let $J \colon \P(\R^d) \to \R$ be the \textit{Fisher information}:
$$J(\rho) = \E_\rho\left[\left\|\nabla \log \rho \right\|^2\right] = -\E_\rho\left[\Delta \log \rho \right],$$
and we define $J(\rho) = +\infty$ if $\rho$ does not have a differentiable density.
The second equality in the definition of $J(\rho)$ above follows by integration by parts assuming the boundary terms are zero.
% if $\rho$ has sufficiently fast decay.

For probability distributions $\rho \ll \nu$ on $\R^d$, the \textit{Kullback-Leibler (KL) divergence} or the \textit{relative entropy} of $\rho$ with respect to $\nu$ is defined by:
$$\KL(\rho\,\|\,\nu) = \E_\rho\left[\log \frac{\rho}{\nu}\right] = \int_{\R^d} \rho(x) \log \frac{\rho(x)}{\nu(x)} \, dx.$$
If $\rho$ and $\nu$ have differentiable density functions, then the {\em relative Fisher information} of $\rho$ with respect to $\nu$ is defined by:
$$\FI(\rho\,\|\,\nu) = \E_\rho\left[\left\|\nabla \log \frac{\rho}{\nu}\right\|^2\right] = \int_{\R^d} \rho(x) \, \left\|\nabla \log \frac{\rho(x)}{\nu(x)}\right\|^2 \, dx.$$
We note that the KL divergence $\KL$ is the Bregman divergence of the negative entropy $-H$, and similarly, the relative Fisher information $\FI$ is the Bregman divergence of the Fisher information $J$.

%%%%%%%%%%%%%%%%%%%%
\subsection{Strong Log-Concavity, Log-Sobolev Inequality, and Poincar\'e Inequality}

We say a probability distribution $\nu \in \P(\R^d)$ is {\em $\alpha$-strongly log-concave} ($\alpha$-SLC) for some $\alpha > 0$ if $-\log \nu \colon \R^d \to \R$ is an $\alpha$-strongly convex function; if $\log \nu$ is twice differentiable, then this is equivalent to $-\nabla^2 \log \nu(x) \succeq \alpha I$ for all $x \in \R^d$.
We say $\nu$ is {\em (weakly) log-concave} if $-\log \nu$ is a convex function, or $-\nabla^2 \log \nu(x) \succeq 0$ for all $x \in \R^d$.

We say $\nu$ is \textit{$L$-log-smooth} for some $0 < L < \infty$ if $-LI \preceq -\nabla^2 \log \nu(x) \preceq LI$ for all $x \in \R^d$.

We say $\nu$ satisfies the {\em $\alpha$-log-Sobolev inequality} ($\alpha$-LSI) for some $\alpha > 0$ if for any probability distribution $\rho$, the following inequality holds:
$$\FI(\rho\,\|\,\nu) \ge 2\alpha \, \KL(\rho\,\|\,\nu).$$
We recall by the Bakry-\'Emery criterion that if $\nu$ is $\alpha$-SLC, then $\nu$ satisfies $\alpha$-LSI~\citep{bakry1985diffusions,villani2009optimal}.
We note that even under SLC, the gap between KL divergence and relative Fisher information can be arbitrarily large; see Lemma~\ref{Lem:LargeFI} in Section~\ref{Sec:LargeFIEx} for an example.

We say $\nu$ satisfies the {\em $\alpha$-Poincar\'e inequality} ($\alpha$-PI) with $\alpha > 0$ if the following inequality holds for any smooth function $\phi \colon \R^d \to \R$:
$$\Var_\nu(\phi) \le \frac{1}{\alpha} \E_\nu[\|\nabla \phi\|^2].$$
Equivalently, for any smooth vector field $\psi \colon \R^d \to \R^d$, we have
$\Var_\nu(\psi) \le \frac{1}{\alpha} \E_\nu[\|\nabla \psi\|_{\HS}^2].$
We recall that if $\nu$ satisfies $\alpha$-LSI, then $\nu$ also satisfies $\alpha$-PI~\citep{villani2009optimal}.

We say a probability distribution $\nu$ on $\R^d$ is {\em symmetric} if $\nu(x) = \nu(-x)$ for all $x \in \R^d$.

%%%%%%%%%%%%%
\subsection{Fokker-Planck Channels}
\label{Sec:FPChannel}

We consider the continuous-time stochastic process $(X_t)_{t \ge 0}$ on $\R^d$ which evolves following the stochastic differential equation (SDE):
\begin{align}\label{Eq:SDE}
    dX_t = b_t(X_t) \, dt + \sqrt{c} \, dW_t
\end{align}
where $b_t \colon \R^d \to \R^d$ is a Lipschitz and twice continuously-differentiable vector field, $c \ge 0$ is a constant, and $(W_t)_{t \ge 0}$ is the standard Brownian motion in $\R^d$ that starts from $W_0 = 0$.
Recall that if $X_t \sim \rho_t$ evolves following the SDE~\eqref{Eq:SDE}, then its probability density function $\rho_t \colon \R^d \to \R$ evolves via the the {\em Fokker-Planck equation}, which is the following partial differential equation (PDE):
\begin{align}\label{Eq:FP}
    \partial_t \rho_t = -\nabla \cdot (\rho_t b_t) + \frac{c}{2} \Delta \rho_t.
\end{align}
We refer to either the SDE~\eqref{Eq:SDE} or the PDE~\eqref{Eq:FP} as the \textit{Fokker-Planck (FP) channel}, which maps any initial random variable $X_0 \sim \rho_0$ to the random variable $X_t \sim \rho_t$ which is the output of the SDE~\eqref{Eq:SDE} at time $t > 0$;
or equivalently, it maps any initial distribution $\rho_0 \in \P(\R^d)$ to the probability distribution $\rho_t \in \P(\R^d)$ which is the solution to the PDE~\eqref{Eq:FP} at time $t > 0$.
Fokker-Planck channels have been studied, e.g., in~\citep{wibisono2017information, WJ18a, WJ18b, zou2025convexity}.

Some examples of the particular FP channels we consider in this paper are the following.

\begin{enumerate}
\item \textbf{Gaussian channel:}
Consider $b_t \equiv 0$, so the SDE is simply a Brownian motion: $dX_t = dW_t$.
The exact solution is:
\begin{align}\label{Eq:GaussianChannel}
    X_t = X_0 + W_t \stackrel{d}{=} X_0 + \sqrt{t} Z
\end{align}
% $X_t = X_0 + W_t \stackrel{d}{=} X_0 + \sqrt{t} Z$,
where $Z \sim \N(0,I)$ is an independent Gaussian random variable in $\R^d$.
This is the classical \textit{Gaussian channel} in information theory.
The Fokker-Planck equation is the {\em heat equation}:
\begin{align}\label{Eq:HeatEq}
    \partial_t \rho_t = \frac{1}{2} \Delta \rho_t.
\end{align}

\item \textbf{Reverse Gaussian channel:}
Let $\pi_t = \pi_0 \, \ast \, \N(0, tI)$ where $\N(0, tI)$ is the Gaussian distribution with mean $0$ and covariance $t I$, so $\pi_t$ satisfies the heat equation: $\partial_t \pi_t = \frac{1}{2} \Delta \pi_t$. 
Let $T > 0$, and define $\nu_t := \pi_{T-t}$, for $0 \le t \le T$.
Then $\nu_t$ satisfies the {\em backward heat equation}:
\begin{align}\label{Eq:BackwardHeatEq}
    \partial_t \nu_t = -\frac{1}{2} \Delta \nu_t = -\nabla \cdot (\nu_t \nabla \log \nu_t) + \frac{1}{2} \Delta \nu_t = -\nabla \cdot (\nu_t \nabla \log \pi_{T-t}) + \frac{1}{2} \Delta \nu_t.
\end{align}
This is the Fokker-Planck equation for the SDE:
\begin{align}\label{Eq:BackwardBM}
    dX_t = \nabla \log \pi_{T-t}(X_t) \, dt + dW_t.
\end{align}
This is a Fokker-Planck channel with time-dependent vector field $b_t = \nabla \log \nu_t = \nabla \log \pi_{T-t}$,
which we call the \textit{reverse Gaussian channel}.
It has the property that if we run the channel~\eqref{Eq:BackwardBM} from $X_0 \sim \rho_0 = \pi_T$, then $X_T \sim \rho_T = \pi_0$, i.e., we ``reverse'' the heat equation.
(However, note that if $\rho_0 \neq \pi_T$, then $\rho_T \neq \pi_0$, so the channel is only suitably constructed for $\pi_0$.)

\item \textbf{Langevin dynamics:}
Let $b_t(x) = -\nabla g(x)$ for some differentiable function $g \colon \R^d \to \R$ with $\int_{\R^d} e^{-g(x)} \, dx < \infty$, and set $c = 2$.
Then the Fokker-Planck channel becomes:
\begin{align}\label{Eq:Langevin}
    dX_t = -\nabla g(X_t) \, dt + \sqrt{2} \, dW_t.
\end{align}
This is the \textit{Langevin dynamics} for sampling from the probability distribution $\nu(x) \propto e^{-g(x)}$.
Indeed, the Fokker-Planck equation becomes:
\begin{align}\label{Eq:FPLangevin}
    \partial_t \rho_t = \nabla \cdot (\rho_t \nabla g) + \Delta \rho_t = \nabla \cdot \left(\rho_t \nabla \log \frac{\rho_t}{\nu} \right)
\end{align}
and we see that $\nu$ is a stationary distribution.

\item \textbf{Ornstein-Uhlenbeck channel:}
Let $b_t(x) = -\gamma x$ for some $\gamma > 0$, and set $c = 2$.
Then the Fokker-Planck channel becomes the \textit{Ornstein-Uhlenbeck (OU) process}:
\begin{align}\label{Eq:OU}
    dX_t = -\gamma X_t \, dt + \sqrt{2} \, dW_t.
\end{align}
This is the Langevin dynamics for the Gaussian target distribution $\nu = \N(0, \frac{1}{\alpha} I)$;
we also call~\eqref{Eq:OU} as the \textit{OU channel}.
The SDE~\eqref{Eq:OU} has an explicit solution at each $t \ge 0$:
\begin{align}\label{Eq:OUSol}
    X_t \stackrel{d}{=} e^{-\gamma t} X_0 + \sqrt{\frac{1-e^{-2\gamma t}}{\gamma}} \, Z
\end{align}
where $Z \sim \N(0,I)$ is an independent Gaussian random variable in $\R^d$.
Note as $\gamma \to 0$, the OU channel~\eqref{Eq:OU} recovers the Gaussian channel~\eqref{Eq:GaussianChannel} at twice the time speed ($2t$ in place of $t$), and the solution~\eqref{Eq:OUSol} recovers the Gaussian channel solution $X_t \stackrel{d}{=} X_0 + \sqrt{2t} \, Z$.
\end{enumerate}

%%%%%%%%%%%%%%%
\subsection{Review of the Convergence Rates along the Langevin Dynamics}
\label{Sec:LangevinReview}

We review the convergence rates of the KL divergence and relative Fisher information along the Langevin dynamics.
This provides a useful context for the formulas and the convergence rates that we find along the more general Fokker-Planck channel.

%%%%%%%%%%%%%%%%
\subsubsection{Convergence rate of the KL divergence}
\label{Sec:LangevinKLRate}

We recall if $(\rho_t)_{t \ge 0}$ evolves following the Fokker-Planck equation~\eqref{Eq:FPLangevin} for the Langevin dynamics~\eqref{Eq:Langevin} targeting $\nu \propto e^{-g}$, then we have the following {\em de Bruijn's identity} on the time derivative of the KL divergence in terms of the relative Fisher information:
\begin{align}\label{Eq:deBruijn}
    \frac{d}{dt} \KL(\rho_t \,\|\, \nu) = - \FI(\rho_t \,\|\, \nu).
\end{align}
This follows by direct time differentiation and integration by parts.
If $\nu$ satisfies $\alpha$-LSI, then from de Bruijn's identity above, along the Langevin dynamics we have $\frac{d}{dt} \KL(\rho_t \,\|\, \nu) \le -2\alpha \KL(\rho_t \,\|\, \nu)$.
Therefore, we conclude an exponential convergence rate of the KL divergence along the Langevin dynamics when the target distribution $\nu$ satisfies LSI:
\begin{align}\label{Eq:LangevinKLRate}
    \KL(\rho_t \,\|\, \nu) \le e^{-2\alpha t} \, \KL(\rho_0 \,\|\, \nu).
\end{align}

%%%%%%%%%%%%%%%%%%%%%
\subsubsection{Convergence rate of the relative Fisher information}

We recall if $(\rho_t)_{t \ge 0}$ evolves following the Fokker-Planck equation~\eqref{Eq:FPLangevin} for the Langevin dynamics~\eqref{Eq:Langevin} targeting $\nu \propto e^{-g}$, then we have the following identity on the time derivative of the relative Fisher information, in terms of the second-order relative Fisher information:
\begin{align}\label{Eq:LangevinFI}
    \frac{d}{dt} \FI(\rho_t\,\|\,\nu) =
    -2\E_{\rho_t}\left[ \left\| \nabla^2 \log \frac{\rho_t}{\nu} \right\|^2_{\HS} \right] -2 \E_{\rho_t}\left[ \left\| \nabla \log \frac{\rho_t}{\nu} \right\|^2_{\nabla^2 g} \right].
\end{align}
This identity follows by direct computation via integration by parts, or via the Otto calculus formula for the Hessian of KL divergence; see~\cite[Ch.~15]{villani2009optimal}.
If $\nu \propto e^{-g}$ is $\alpha$-SLC, so $\nabla^2 g(x) \succeq \alpha I$ for all $x \in \R^d$, then from the identity~\eqref{Eq:LangevinFI}, along the Langevin dynamics we have
$\frac{d}{dt} \FI(\rho_t\,\|\,\nu) \le -2\alpha \FI(\rho_t\,\|\,\nu)$.
Therefore, we conclude an exponential convergence rate of the relative Fisher information along the Langevin dynamics when the target distribution $\nu$ is $\alpha$-strongly log-concave:
\begin{align}\label{Eq:LangevinFIRate}
    \FI(\rho_t\,\|\,\nu) \le e^{-2\alpha t} \, \FI(\rho_0\,\|\,\nu).
\end{align}
We can verify that this convergence rate is tight in the Gaussian case.

\begin{example}[Convergence rate of the relative Fisher information along the OU process.] \label{Ex:LangevinRate}
    Consider the OU process~\eqref{Eq:OU}, which is the Langevin dynamics~\eqref{Eq:Langevin} for the target Gaussian distribution $\nu = \N(0, \gamma^{-1} I)$, which is $\gamma$-strongly log-concave.
    Suppose we start the OU process from $\rho_0 = \N(m, s^{-1} I)$ for some $m \in \R^d$ and $s > 0$.
    For all $t \ge 0$, the solution~\eqref{Eq:OUSol} to the OU channel is $\rho_t = \N(e^{-\gamma t} m, s_t^{-1} I)$ where 
    $s_t^{-1} = e^{-2\gamma t} s^{-1} + (1-e^{-2\gamma t}) \gamma^{-1}.$
    Note that $\min\{s,\gamma\} \le s_t \le \max\{s,\gamma\}$, so $s_t = \Theta(1)$ for all $t \ge 0$.
    Furthermore, note that
    $(s_t-\gamma)^2 = s_t^2 \gamma^2 (s_t^{-1}-\gamma^{-1})^2 = e^{-4\gamma t} s_t^2 \gamma^2 (s^{-1}-\gamma^{-1})^2.$
    Then we can compute:
    \begin{align*}
        \FI(\rho_t \,\|\, \nu) 
        &= e^{-2\gamma t} \, \gamma^2 \|m\|^2 + \frac{d}{s_t} \, (s_t-\gamma)^2 \\
        &= e^{-2\gamma t} \, \gamma^2 \|m\|^2 + e^{-4\gamma t} \, d \gamma^2 s_t \, (s^{-1}-\gamma^{-1})^2.
    \end{align*}
    We observe that if $m \neq 0$, then $\FI(\rho_t \,\|\, \nu) = O(e^{-2\gamma t})$, as predicted in the general convergence rate~\eqref{Eq:LangevinFIRate}.
    (Note that if $m  = 0$, then the rate improves to $\FI(\rho_t \,\|\, \nu) = O(e^{-4\gamma t})$;
    this is because in this case $\rho_t$ and $\nu$ are symmetric, and $\rho_t$ satisfies a Poincar\'e inequality, so we can use the second-order relative Fisher information in~\eqref{Eq:LangevinFI} to improve the convergence rate, 
    see Theorem~\ref{Thm:FisherInfoOUFlow}(ii).)
\end{example}

%%%%%%%%%%%%%%%%%%
\subsection{Review of the SDPI in KL Divergence along the Fokker-Planck Channel}
\label{Sec:ReviewSDPI}

We recall the following formula for the time derivative of the KL divergence along simultaneous applications of the Fokker-Planck channel; this implies the SDPI in KL divergence along the Fokker-Planck channel when the second distribution satisfies LSI.
We note the similarity with de Bruijn's identity~\eqref{Eq:deBruijn}, which is a special case of the identity~\eqref{Eq:ddtKLGeneral} below for the Langevin dynamics when the second distribution is stationary.
For completeness, we provide the proof of Lemma~\ref{Lem:TimeDerivativeKL} in Section~\ref{Sec:TimeDerivativeKLProof}.

\begin{lemma}\label{Lem:TimeDerivativeKL}
    Suppose $(\rho_t)_{t \ge 0}$ and $(\nu_t)_{t \ge 0}$ evolve following the Fokker-Planck channel~\eqref{Eq:FP}:
    \begin{align*}
        \partial_t \rho_t &= -\nabla \cdot (\rho_t b_t) + \frac{c}{2} \Delta \rho_t \enspace,  \\
        % \qquad\qquad
        \partial_t \nu_t &= -\nabla \cdot (\nu_t b_t) + \frac{c}{2} \Delta \nu_t.
    \end{align*}
    Then for any $t \ge 0$, we have:
    \begin{align}\label{Eq:ddtKLGeneral}
        \frac{d}{dt} \KL(\rho_t \,\|\, \nu_t) 
        &= -\frac{c}{2} \FI(\rho_t \,\|\, \nu_t).
    \end{align}
    In particular, if for all $t \ge 0$ we know that $\nu_t$ satisfies $\alpha_t$-LSI, then:
    \begin{align}\label{Eq:SDPI_KL}
        \KL(\rho_t \,\|\, \nu_t) \le \exp\left(-c \int_0^t \alpha_s \, ds \right) \KL(\rho_0 \,\|\, \nu_0).
    \end{align}
\end{lemma}

A similar formula to~\eqref{Eq:ddtKLGeneral} holds for all $\Phi$-divergence and R\'enyi divergence~\citep{CCSW22,kook2024inandout,mitra2025fast};
we only present it for KL divergence to draw the comparison with the relative Fisher information that we study in this paper.
% The formula above has been useful in many previous works.
In the case of the heat flow~\eqref{Eq:HeatEq}, the formula~\eqref{Eq:ddtKLGeneral} can be used to show an SDPI along the Gaussian channel under LSI~\citep{W11}.
Generalizations of the formula above have been used to prove the mixing time guarantees of ULA and the Proximal Sampler in~\citep{VW19,CCSW22,YLFWC23,LMW24,kook2024inandout,mitra2025fast}.
See~\cite{klartag25heat} for a rigorous derivation of the identity~\eqref{Eq:ddtKLGeneral} above in the case of the heat flow.

%%%%%%%%%%%%%%%%%%
\section{Behavior of the Relative Fisher Information along the Fokker-Planck Channel}
\label{Sec:ContractionFI}

We present results on the contraction of relative Fisher information along Fokker-Planck channels.

%%%%%%%%%%%%%%%%%%
\subsection{Time Derivative of the Relative Fisher Information along the Fokker-Planck Channel}

Our key tool is the following formula for the time derivative of the relative Fisher information along simultaneous Fokker-Planck channel~\eqref{Eq:FP}.
This formula is the Fisher information analogue of the identity~\eqref{Eq:ddtKLGeneral} for KL divergence.
We provide the proof of Lemma~\ref{Lem:TimeDerivativeFI} in Section~\ref{Sec:TimeDerivativeFIProof}, 
assuming the distributions satisfy regularity properties for the operations in our proof to hold.
We note that~\cite[Conjecture~4.2]{klartag25heat} provides a time derivative formula for the relative $\Phi$-Fisher information along heat flow, which agrees with the formula~\eqref{Eq:TimeDerivativeFI} when $b_t = 0$ and $\Phi(x) = x \log x$.

\begin{lemma}\label{Lem:TimeDerivativeFI}
    Suppose $(\rho_t)_{t \ge 0}$ and $(\nu_t)_{t \ge 0}$ evolve following the Fokker-Planck equation~\eqref{Eq:FP}:
    \begin{align*}
        \partial_t \rho_t &= -\nabla \cdot (\rho_t b_t) + \frac{c}{2} \Delta \rho_t \enspace,  \\
        % \qquad\qquad
        \partial_t \nu_t &= -\nabla \cdot (\nu_t b_t) + \frac{c}{2} \Delta \nu_t \enspace.
    \end{align*}
    Then for any $t \ge 0$:
    \begin{align}\label{Eq:TimeDerivativeFI}
        \frac{d}{dt} \FI(\rho_t \,\|\, \nu_t) 
        = -c \, \E_{\rho_t} \left[\left\| \nabla^2 \log \frac{\rho_t}{\nu_t} \right\|^2_{\HS} \right] 
        - 2\E_{\rho_t}\left[\left\|\nabla \log \frac{\rho_t}{\nu_t} \right\|^2_{(-c \, \nabla^2 \log \nu_t + (\nabla b_t)_{\sym})}\right].
    \end{align}
\end{lemma}

We consider a few special cases of interest:
\begin{enumerate}
    \item \textbf{Langevin dynamics:} Consider when $b_t = -\nabla g$ and $c = 2$, so the channel is the Langevin dynamics~\eqref{Eq:Langevin} for target distribution $\nu \propto e^{-g}$.
    Suppose that $\nu_0 = \nu$, so $\nu_t = \nu$ for all $t \ge 0$.
    In this case, $-c\nabla^2 \log \nu_t + (\nabla b_t)_{\sym} = 2\nabla^2 g - \nabla^2 g = \nabla^2 g$, so the identity~\eqref{Eq:TimeDerivativeFI} recovers the time derivative formula~\eqref{Eq:LangevinFI} along the Langevin dynamics from Section~\ref{Sec:LangevinReview}.
    \item \textbf{Gaussian channel:} When $b_t = 0$ and $c=1$, the identity~\eqref{Eq:TimeDerivativeFI} implies (S)DPIs in relative Fisher information for the Gaussian channel under (strong) log-concavity assumptions on $\nu_0$; see Section~\ref{Sec:SDPIGaussian}.
    \item \textbf{Reverse Gaussian channel:} Consider when $b_t = \nabla \log \nu_t$ and $c=1$, so the Fokker-Planck channel is the reverse Gaussian channel~\eqref{Eq:BackwardBM}.
    In this case, $-c\nabla^2 \log \nu_t + (\nabla b_t)_{\sym} = 0$. 
    Thus, DPI always holds along the reverse Gaussian channel.
    We use this property to prove the mixing time guarantee of the Proximal Sampler in relative Fisher information; see Section~\ref{Sec:ProximalSampler}.
    \item \textbf{OU channel:} When $b_t(x) = -\gamma x$ and $c=2$, the identity~\eqref{Eq:TimeDerivativeFI} implies SDPI in relative Fisher information eventually holds for the OU channel~\eqref{Eq:OU} under strong log-concavity assumption on $\nu_0$; see Section~\ref{Sec:SDPI-OU}.
\end{enumerate}

%%%%%%%%%%%%%%%%
\subsection{SDPI in Relative Fisher Information along the Gaussian Channel}
\label{Sec:SDPIGaussian}

In Theorem~\ref{Thm:FisherInfoHeatFlow} below we show for relative Fisher information along the Gaussian channel: (i) DPI holds when the second distribution is log-concave; 
(ii) SDPI holds when the second distribution is SLC;
(iii) SDPI with an improved rate holds when the first distribution additionally satisfies Poincar\'e inequality and symmetry; and
(iv) SDPI eventually holds when the second distribution is a log-Lipschitz perturbation of an SLC distribution.
The proofs use the time derivative formula from Lemma~\ref{Lem:TimeDerivativeFI}; for part (iv), we also invoke the result of~\citep{brigati2024heat} on the evolution of log-concavity along the heat flow.
We provide the proof of Theorem~\ref{Thm:FisherInfoHeatFlow} in Section~\ref{Sec:FisherInfoHeatFlowProof}.

\begin{theorem}\label{Thm:FisherInfoHeatFlow}
    Let $\rho_0, \nu_0$ be probability distributions on $\R^d$ with $\FI(\rho_0 \,\|\, \nu_0) < \infty$.
    Let $\rho_t, \nu_t$ be the output of the Gaussian channel~\eqref{Eq:GaussianChannel} from $\rho_0, \nu_0$, respectively, at time $t > 0$.
    Then we have:
    \begin{enumerate}
        \item[(i)] Assume $\nu_0$ is log-concave. Then for all $t \ge 0$:~
        $$\FI(\rho_t \,\|\, \nu_t) \le \FI(\rho_0 \,\|\, \nu_0).$$
        \item[(ii)] Assume $\nu_0$ is $\alpha$-SLC for some $\alpha > 0$. Then for all $t \ge 0$:~
        $$\FI(\rho_t \,\|\, \nu_t) \le \dfrac{\FI(\rho_0 \,\|\, \nu_0)}{(1+\alpha t)^2}.$$
        \item[(iii)] Assume $\nu_0$ is $\alpha$-SLC and $\rho_0$ satisfies $\beta$-PI for some $\alpha,\beta > 0$.
        Assume $\rho_0, \nu_0$ are symmetric.
        Then for all $t \ge 0$:~
        $$\FI(\rho_t \,\|\, \nu_t) \le \dfrac{\FI(\rho_0 \,\|\, \nu_0)}{(1+\beta t)(1+\alpha t)^2}.$$
        \item[(iv)] Assume $\nu_0 \propto e^{-g-\psi}$ where $g \colon \R^d \to \R$ is $\alpha$-strongly convex and $\psi \colon \R^d \to \R$ is $L$-Lipschitz for some $\alpha > 0$, $0 \le L < \infty$. Then for all $t \ge 0$:
        \begin{align}\label{Eq:FIHeatFlowEventualSDPI}
            \FI(\rho_t \,\|\, \nu_t) \le \frac{\FI(\rho_0 \,\|\, \nu_0)}{(1+\alpha t)^2} \exp\left(\frac{2tL^2}{\alpha t+1} + \frac{8L \sqrt{t}}{\sqrt{\alpha t+1}}\right).
        \end{align}  
    \end{enumerate}        
\end{theorem}

We can verify the rate in part (ii) is tight by a calculation in the Gaussian case, and that we get an improved rate under Poincar\'e and symmetry as in part (iii), see Example~\ref{Ex:GaussianHeat} in Section~\ref{Sec:SDPIGaussianExample}.

In Theorem~\ref{Thm:FisherInfoHeatFlow} part (iv) above, we note that the upper bound~\eqref{Eq:FIHeatFlowEventualSDPI} is initially increasing for small $t > 0$, and it is decreasing for large $t$, so it shows we eventually have SDPI.
Furthermore, for all $t \ge 0$, we can further upper bound~\eqref{Eq:FIHeatFlowEventualSDPI} by:
$\FI(\rho_t \,\|\, \nu_t) \le C \cdot \frac{\FI(\rho_0 \,\|\, \nu_0)}{(1+\alpha t)^2}$ where $C =  \exp\left(\frac{2L^2}{\alpha} + \frac{8L}{\sqrt{\alpha}}\right)$.
We can construct an explicit example where the relative Fisher information along the heat flow indeed exhibits this behavior of initial increase and eventual decrease, see Example~\ref{Ex:CounterDecrHeat} below.

%%%%%%%%%%%%%%%%%
\subsubsection{Example of SDPI in Fisher information along Gaussian channel}
\label{Sec:SDPIGaussianExample}

\begin{example}\label{Ex:GaussianHeat}
Let $\rho_0 = \N(m, s I)$ and $\nu_0 = \N(0,I)$ for some $m \in \R^d$ and $s > 0$, $s \neq 1$.
Note that $\nu_0$ is $1$-SLC, $\nu_0$ is symmetric, and $\rho_0$ satisfies $(1/s)$-PI.
Along the Gaussian channel~\eqref{Eq:GaussianChannel}, we know
$\rho_t = \N(m, (s+t) I)$ and $\nu_t = \N(0, (1+t) I)$.
Then we can compute:
$$\FI(\rho_t \,\|\, \nu_t) = \frac{(s-1)^2 d}{(1+t)^2 (s+t)^2} + \frac{\|m\|^2}{(1+t)^2}.$$
We note: (1) If $m \neq 0$, then $\FI(\rho_t \,\|\, \nu_t) = O(t^{-2})$ as predicted in Theorem~\ref{Thm:FisherInfoHeatFlow}(ii). 
(2) If $m = 0$, then $\rho_0$ is symmetric, so Theorem~\ref{Thm:FisherInfoHeatFlow}(iii) predicts an improved rate of $O(t^{-3})$; in this example, we actually have a faster rate of $\FI(\rho_t \,\|\, \nu_t) = O(t^{-4})$.
\end{example}

We remark that the faster rate $\FI(\rho_t \,\|\, \nu_t) = O(t^{-4})$ in the case $m=0$ above is likely because $\rho_0$ satisfies an improved Poincar\'e inequality for symmetric test functions, which follows from the improved LSI result by~\citep{fathi2016quantitative}.

%%%%%%%%%%%%%%%%%%%%%%
\subsubsection{Counterexample of DPI in Fisher information along Gaussian channel}

We show DPI for relative Fisher information does \textit{not} always hold even for the Gaussian channel.
Note by Theorem~\ref{Thm:FisherInfoHeatFlow}(i), the counterexample needs the second distribution $\nu_0$ to be non-log-concave.

%%%%%%%%%%%
\begin{example}\label{Ex:CounterDecrHeat}
    Let $\rho_0 = \N(0,1)$ on $\R$ ($d=1$).
    Let $\nu_0 \propto e^{-g}$ where $g \colon \R \to \R$ is defined by:
    \begin{align*} 
        g(x) = 
        \begin{cases}
            -\frac{M}{2} x^2 ~~~ & \text{ if } |x| \le L, \\
           \frac{1}{2} (x-L)^2 - M L(x-L) - \frac{M L^2}{2}  & \text{ if } x > L, \\
            \frac{1}{2} (x+L)^2 + M L(x+L) - \frac{M L^2}{2}  & \text{ if } x < -L,
        \end{cases}
    \end{align*}
    where $M, L \ge 2$ are arbitrary.
    Then along the Gaussian channel~\eqref{Eq:GaussianChannel}, for small $t > 0$, the relative Fisher information is increasing:
    $\FI(\rho_t\,\|\,\nu_t) > \FI(\rho_0\,\|\,\nu_0).$
    We prove this in Proposition~\ref{Prop:CounterexampleDecayGaussianChannel} in Section~\ref{Sec:CounterexampleDecayGaussianChannelProof}.
    See Figure~\ref{Fig:CounterexampleDecreaseHeat} for an illustration. 
\end{example}

\begin{figure}[h!t!]
    \centering
    \subfigure[Density of $\nu_0$]{\includegraphics[width=0.32\textwidth]{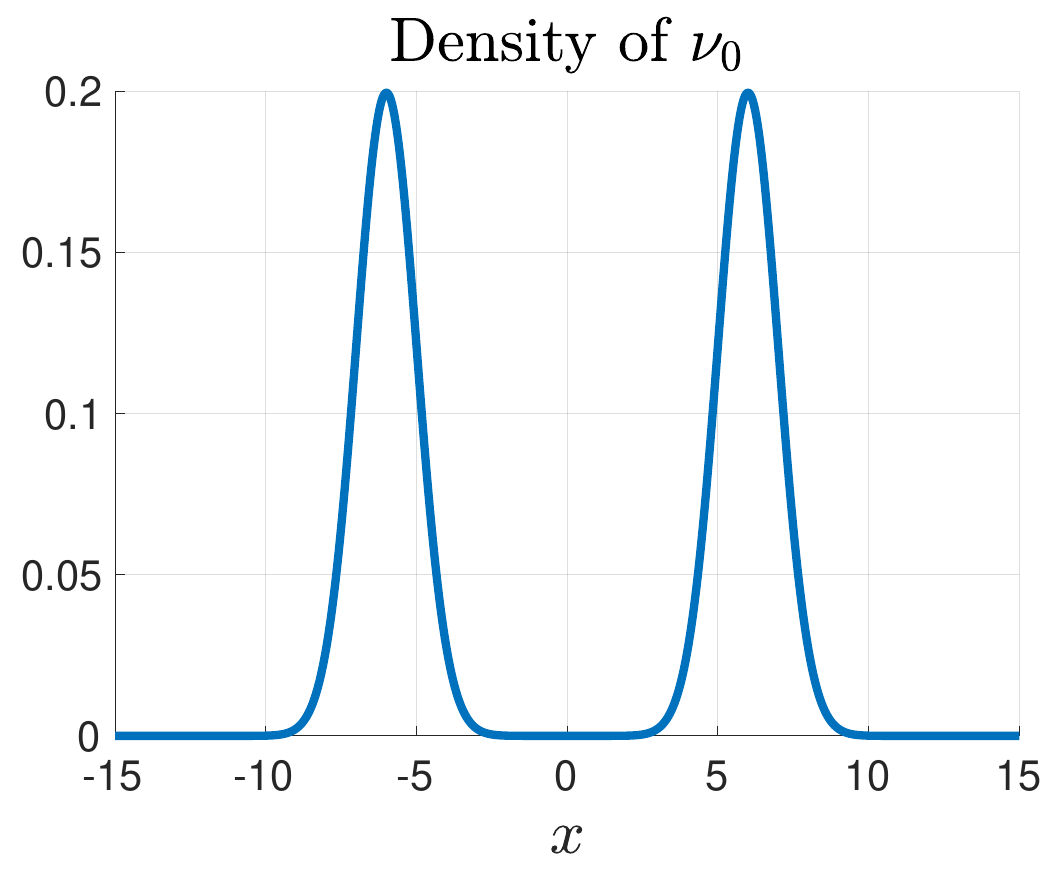}}
    \subfigure[Plot of $\FI(\rho_t\,\|\,\nu_t)$]{\includegraphics[width=0.32\textwidth]{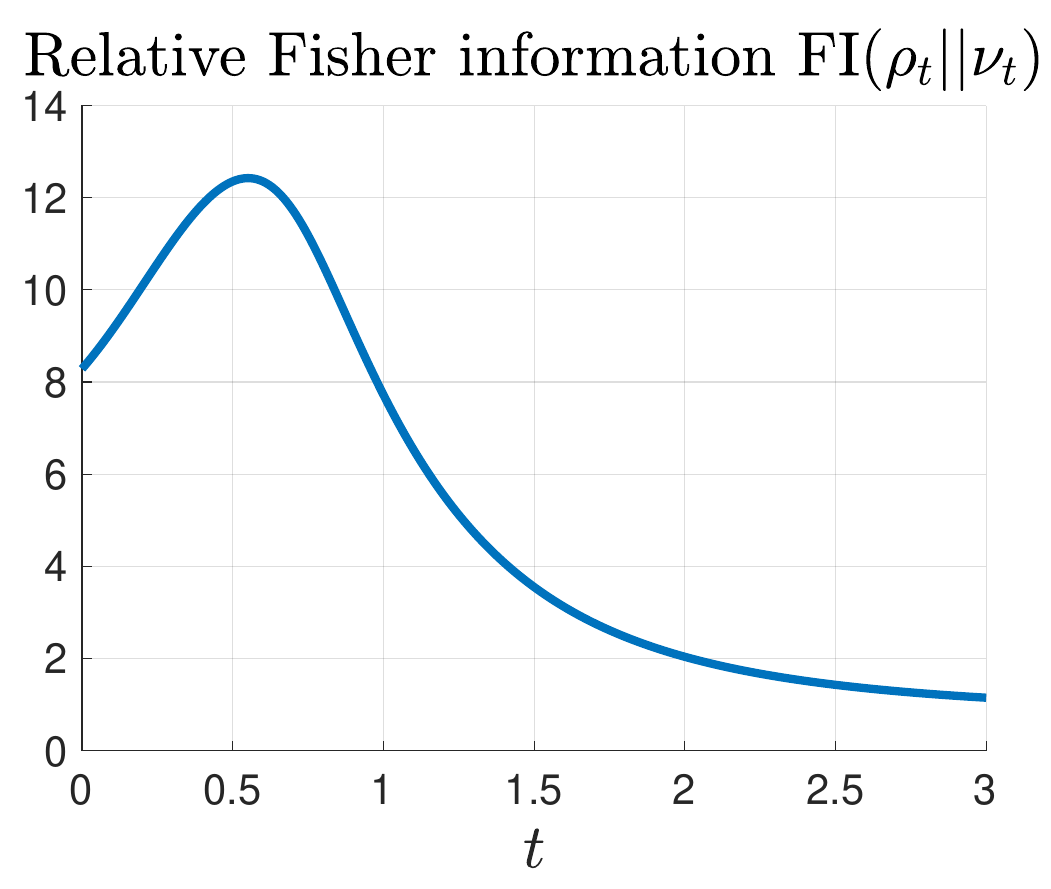}}
    \subfigure[Plot of $\KL(\rho_t\,\|\,\nu_t)$]{\includegraphics[width=0.32\textwidth]{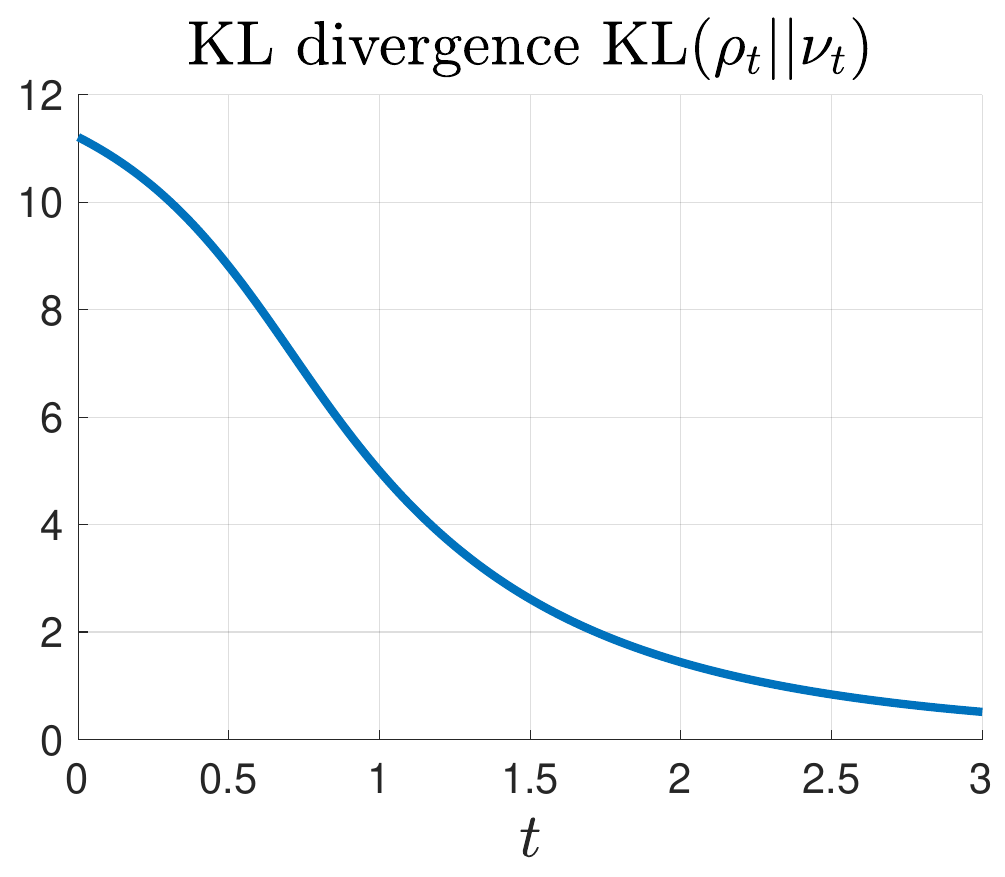}}
    \caption{An example where DPI in relative Fisher information does not hold along the Gaussian channel; 
    see Example~\ref{Ex:CounterDecrHeat} and Proposition~\ref{Prop:CounterexampleDecayGaussianChannel} for details.
    }
    \label{Fig:CounterexampleDecreaseHeat}
\end{figure}

Figure~\ref{Fig:CounterexampleDecreaseHeat} shows an illustration of Example~\ref{Ex:CounterDecrHeat} above.
Figure~\ref{Fig:CounterexampleDecreaseHeat}(a) shows the density of $\nu_0$ as defined in Proposition~\ref{Prop:CounterexampleDecayGaussianChannel} when $M = L = 2$, with $\rho_0 = \N(0,1)$. Figure~\ref{Fig:CounterexampleDecreaseHeat}(b) shows the plot of $t \mapsto \FI(\rho_t \,\|\, \nu_t)$, which is initially increasing (hence DPI does not hold), then eventually decreasing.
Recall $\frac{d}{dt} \KL(\rho_t \,\|\, \nu_t) = -\frac{1}{2} \FI(\rho_t \,\|\, \nu_t)$ along the Gaussian channel.
This means the KL divergence (which is always decreasing along the Gaussian channel by classical DPI) is initially decreasing in a concave way, then eventually in a convex way, as shown in Figure~\ref{Fig:CounterexampleDecreaseHeat}(c).

%%%%%%%%%%%%%%%%%%%%%%%
\subsection{Eventual SDPI in Relative Fisher Information along the OU Channel}
\label{Sec:SDPI-OU}

We derive the following results on the eventual SDPI for the relative Fisher information along OU channel~\eqref{Eq:OU} under strong log-concavity.
In Theorem~\ref{Thm:FisherInfoOUFlow} below we show along the OU channel:
(i) SDPI for relative Fisher information eventually holds with rate $O(e^{-2\gamma t})$ when the second distribution is strongly log-concave; and
(ii) SDPI eventually holds with an improved rate $O(e^{-4\gamma t})$ when the first distribution additionally satisfies Poincar\'e inequality and both distributions are symmetric.
We provide the proof of Theorem~\ref{Thm:FisherInfoOUFlow} in Section~\ref{Sec:SDPI-OU-Proof}.
However, even when both distributions are Gaussian, the relative Fisher information can be initially increasing (see Figure~\ref{Fig:SDPIOUEx} and Example~\ref{Ex:SDPIOU}).

\begin{theorem}\label{Thm:FisherInfoOUFlow}
    Let $\rho_0, \nu_0$ be probability distributions on $\R^d$ with $\FI(\rho_0 \,\|\, \nu_0) < \infty$.
    Along the OU channel~\eqref{Eq:OU} for target distribution $\N(0, \gamma^{-1} I)$, let $\rho_t, \nu_t$ be the output distributions at time $t > 0$ from input distributions $\rho_0, \nu_0$, respectively.
    Then we have:
    \begin{enumerate}
        \item[(i)] Assume $\nu_0$ is $\alpha$-SLC for some $\alpha > 0$.
        Then for all $t \ge 0$:
        $$\FI(\rho_t \,\|\, \nu_t) \le \frac{\gamma^2 \, e^{-2\gamma t}}{(\alpha + e^{-2\gamma t}(\gamma-\alpha))^2} \, \FI(\rho_0 \,\|\, \nu_0).$$
        \item[(ii)] Assume $\nu_0$ is $\alpha$-SLC and $\rho_0$ satisfies $\beta$-PI for some $\alpha, \beta > 0$.
        Assume $\rho_0$, $\nu_0$ are symmetric.
        Then for all $t \ge 0$:
        $$\FI(\rho_t \,\|\, \nu_t) \le \frac{\gamma^3 \, e^{-4\gamma t}}{(\beta + e^{-2\gamma t} (\gamma-\beta))(\alpha+e^{-2 \gamma t}(\gamma-\alpha))^2} \, \FI(\rho_0 \,\|\, \nu_0).$$
    \end{enumerate}        
\end{theorem}

We make a few remarks:
Recall that as $\gamma \to 0$, the OU channel recovers the Gaussian channel~\eqref{Eq:GaussianChannel} at twice the time speed.
We note that as $\gamma \to 0$, the upper bound in Theorem~\ref{Thm:FisherInfoOUFlow}(i) recovers the upper bound for the Gaussian channel in Theorem~\ref{Thm:FisherInfoHeatFlow}(ii) at time $2t$, and the upper bound in Theorem~\ref{Thm:FisherInfoOUFlow}(ii) recovers the upper bound for the Gaussian channel in Theorem~\ref{Thm:FisherInfoHeatFlow}(iii) at time $2t$.

We also recall from Theorem~\ref{Thm:FisherInfoHeatFlow}(i) that for the Gaussian channel, DPI for relative Fisher information holds when the second distribution is weakly log-concave.
In contrast, for the OU channel, we are unable to show DPI when the second distribution is only weakly log-concave.
By taking $\alpha \to 0$ in Theorem~\ref{Thm:FisherInfoOUFlow}(i), the upper bound becomes $e^{2\gamma t}$ which is increasing in $t$ since $\gamma > 0$, so we cannot conclude DPI
(but as we further take $\gamma \to 0$, this upper bound becomes $1$, which agrees with the DPI in Theorem~\ref{Thm:FisherInfoHeatFlow}(i) as the OU channel recovers the Gaussian channel).

%%%%%%%%%%%%%%%%
\subsubsection{Example of eventual SDPI along the OU channel}

We can verify that the contraction rates in Theorem~\ref{Thm:FisherInfoOUFlow} are tight for large $t$ by explicit calculation in the Gaussian case; see Example~\ref{Ex:SDPIOU}.
However, we also note that even when both distributions are strongly log-concave, SDPI for relative Fisher information may not hold for small $t$; see Figure~\ref{Fig:SDPIOUEx}.

%%%%%%%%%%%%%%%%%%
\begin{example}[Eventual SDPI along the OU channel.] \label{Ex:SDPIOU}
Let $\rho_0 = \N(m, \beta^{-1} I)$ and $\nu_0 = \N(0, \alpha^{-1} I)$ for some $m \in \R^d$ and $\alpha,\beta > 0$.
Note that $\nu_0$ is $\alpha$-SLC and symmetric, and $\rho_0$ is $\beta$-SLC so it satisfies $\beta$-PI.
Along the OU channel~\eqref{Eq:OU} to target distribution $\N(0, \gamma^{-1} I)$, the solutions are:
\begin{align*}
    \rho_t &= \N(e^{-\gamma t} m, \; \beta_t^{-1} I), 
    \qquad\qquad
    \nu_t = \N(0, \; \alpha_t^{-1} I)
\end{align*}
where $\beta_t^{-1} := e^{-2\gamma t} \beta^{-1} + (1-e^{-2\gamma_t}) \gamma^{-1}$, and $\alpha_t^{-1} := e^{-2\gamma t} \alpha^{-1} + (1-e^{-2\gamma_t}) \gamma^{-1}$.
Note $\min\{\beta,\gamma\} \le \beta_t \le \max\{\beta,\gamma\}$ and $\min\{\alpha,\gamma\} \le \alpha_t \le \max\{\alpha,\gamma\}$, so $\beta_t = \Theta(1)$ and $\alpha_t = \Theta(1)$ for all $t \ge 0$.
Furthermore,
$\frac{1}{\beta_t} \left(\alpha_t - \beta_t\right)^2 
= \alpha_t^2 \beta_t \left(\beta_t^{-1} - \alpha_t^{-1} \right)^2 
= e^{-4\gamma t} \, \alpha_t^2 \beta_t \left(\beta^{-1} - \alpha^{-1} \right)^2.$
Then we can compute:
\begin{align}
    \FI(\rho_t\,\|\,\nu_t) 
    &= e^{-2\gamma t} \, \beta_t^2  \|m\|^2 + \frac{d}{\beta_t} \left(\alpha_t - \beta_t\right)^2 \notag \\
    &= e^{-2\gamma t} \, \beta_t^2  \|m\|^2 + e^{-4\gamma t} \, d \alpha_t^2 \beta_t \left(\beta^{-1} - \alpha^{-1} \right)^2. \label{Eq:FIOUGaussian} 
\end{align}
We observe for large $t$: (1) If $m \neq 0$, then $\FI(\rho_t \,\|\, \nu_t) = O(e^{-2\gamma t})$ as predicted in Theorem~\ref{Thm:FisherInfoOUFlow}(i); and 
(2)~If $m = 0$, then $\rho_0$ is symmetric, and $\FI(\rho_t \,\|\, \nu_t) = O(e^{-4\gamma t})$ as predicted in Theorem~\ref{Thm:FisherInfoOUFlow}(ii).

However, note that for small $t > 0$, the value~\eqref{Eq:FIOUGaussian} can be decreasing.
Consider in $d=1$ dimension when $m=0$, $\gamma = 1$, $\beta = 100$, and $\alpha = 0.1$.
See the plot in Figure~\ref{Fig:SDPIOUEx}, which shows the relative Fisher information is initially increasing.
\end{example}

\begin{figure}[h!t!]
    \centering
    \includegraphics[width=0.4\textwidth]{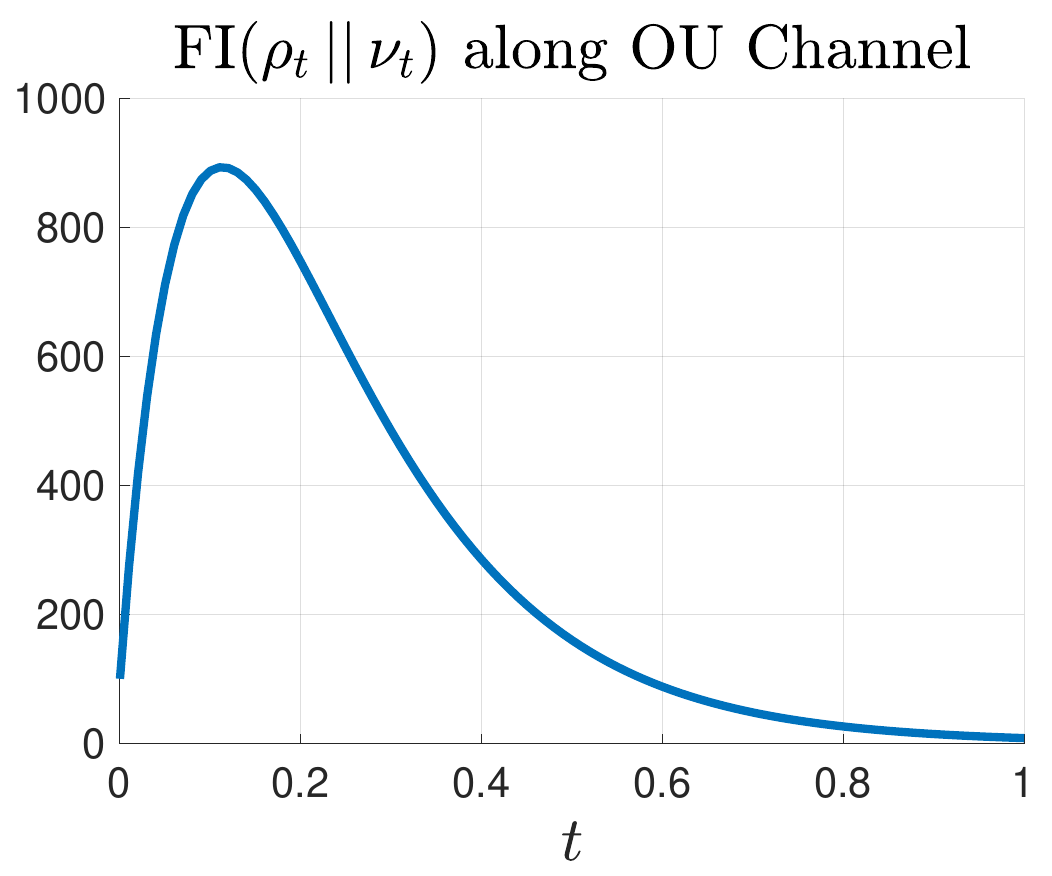}
    \caption{An example where DPI in relative Fisher information does not hold along the OU channel.
    Here $\rho_0 = \N(0,0.01)$, $\nu_0 = \N(0,10)$, and the OU channel is targeting $\N(0,1)$.
    See Example~\ref{Ex:SDPIOU}.
    % for details.
    }
    \label{Fig:SDPIOUEx}
\end{figure}

We also note that since the initial behaviors of the OU channel and the Gaussian channel are similar, the same example from Example~\ref{Ex:CounterDecrHeat} (for the counterexample for DPI along the Gaussian channel) also shows that the relative Fisher information is initially increasing along the OU channel; this can be verified by the same calculation as in the proof of Proposition~\ref{Prop:CounterexampleDecayGaussianChannel}, which we omit.

%%%%%%%%%%%%%%%%%
\section{Analysis of the Proximal Sampler}
\label{Sec:ProximalSamplerAnalysis}

We review the Proximal Sampler in Section~\ref{Sec:ProximalSamplerReview} and its convergence guarantee in KL divergence in Section~\ref{Sec:ProximalSamplerReviewKL}.
We present our result on the convergence guarantee of the Proximal Sampler in relative Fisher information in Section~\ref{Sec:ProximalSampler}.

%%%%%%%%%%%%%%%%
\subsection{Review of the Proximal Sampler}
\label{Sec:ProximalSamplerReview}

Suppose we want to sample from a target probability distribution $\nu^X(x) \propto e^{-g(x)}$ on $\R^d$. 
We define a target joint distribution $\nu^{XY}$ on the extended state space $\R^{2d} = \R^d \times \R^d$ with density function:
\begin{align}\label{Eq:nuXY}
    \nu^{XY}(x,y) \propto \exp\left(-g(x) - \frac{1}{2\eta}\|y-x\|^2\right)
\end{align}
where $\eta > 0$ is the \textit{step size}.
Observe that the $X$-marginal of $\nu^{XY}$ is equal to $\nu^X$.
Therefore, if we can sample $(X,Y) \sim \nu^{XY}$, then by ignoring the $Y$ component, we obtain a sample $X \sim \nu^X$.

The {\em Proximal Sampler}~\citep{LST21} for sampling from $\nu^X$ performs Gibbs sampling to sample from the joint distribution $\nu^{XY}$.
In each iteration $k \ge 0$, from the current iterate $x_k \sim \rho_k^X$, the Proximal Sampler performs the following two steps:
\begin{enumerate}
    \item (\textit{Forward step}:) Sample $y_k \mid x_k \sim \nu^{Y \mid X}(\cdot \mid x_k) = \N(x_k, \eta I)$.
    This results in a new iterate $y_k \sim \rho_k^Y$, where $\rho_k^Y = \rho_k^X \ast \N(0, \eta I)$.
    \item (\textit{Backward step}:) Sample $x_{k+1} \mid y_k \sim \nu^{X \mid Y}(\cdot \mid y_k)$.
    This gives the next iterate $x_{k+1} \sim \rho_{k+1}^X$.
\end{enumerate}
Since Gibbs sampling is a reversible Markov chain, $\nu^{XY}$ is a stationary distribution of the forward and backward steps above.
Thus, the Proximal Sampler is {\em unbiased}: If $\rho_k^X = \nu^X$, then $\rho_{k+1}^X = \nu^X$.

We note the forward step is easy to implement by adding an independent Gaussian noise.
The backward step requires sampling from the conditional distribution $\nu^{X \mid Y}$ with density:
\begin{align}\label{Eq:RGO}
    \nu^{X \mid Y}(x \mid y_k) \propto_x \exp\left(-g(x) - \frac{1}{2\eta}\|x-y_k\|^2\right).
\end{align}
We assume we can sample from $\nu^{X \mid Y}$~\eqref{Eq:RGO}; this is called the {\em Restricted Gaussian Oracle (RGO)} in~\citep{LST21}.
When $\nu^X$ is log-smooth (i.e., when $g$ has bounded second derivatives), for sufficiently small step size $\eta$, we can implement the RGO via rejection sampling with an $O(1)$ expected number of queries to $g$; see Section~\ref{Sec:RejectionReview} for a review.

%%%%%%%%%%%%
\subsection{Review of the Mixing Time of Proximal Sampler in KL divergence}
\label{Sec:ProximalSamplerReviewKL}

We can interpret each step of the Proximal Sampler as the evolution along a Fokker-Planck channel:
\begin{enumerate}
    \item The forward step is the result of passing $\rho_k^X$ via the Gaussian channel~\eqref{Eq:GaussianChannel} for time $\eta$ to obtain $\rho_k^Y = \rho_k^X \ast \N(0,\eta I)$.
    We also apply the forward step to $\nu^X$ to obtain $\nu^Y = \nu^X \ast \N(0,\eta I)$.

    \item The backward step is the result of passing $\rho_k^Y$ via the reverse Gaussian channel~\eqref{Eq:BackwardBM} for time $\eta$ to obtain $\rho_{k+1}^X$.
    When we apply the backward step to $\nu^Y$, we obtain $\nu^X$ back.

\end{enumerate}

By using the interpretations above and invoking the SDPI for the forward and reverse Gaussian channels, \citep{CCSW22} show the following mixing time for Proximal Sampler in KL divergence under LSI, where each forward and backward step induces a contraction in KL divergence.

\begin{lemma}[{\cite[Theorem~3]{CCSW22}}]
    Assume $\nu^X$ satisfies $\alpha$-LSI for some $\alpha > 0$.
    Along each iteration of the Proximal Sampler with any $\eta > 0$, we have:
    \begin{subequations}\label{Eq:KLProximalRate}
        \begin{align}
            \KL(\rho_k^Y \,\|\, \nu^Y) \,&\le\, \frac{\KL(\rho_k^X \,\|\, \nu^X)}{(1+\alpha \eta)}
            \label{Eq:KLProximalRate-1} \\
            % \qquad\qquad
            \KL(\rho_{k+1}^X \,\|\, \nu^X) \,&\le\, \frac{\KL(\rho_k^Y \,\|\, \nu^Y)}{(1+\alpha \eta)}. \label{Eq:KLProximalRate-2}
        \end{align}        
    \end{subequations}
    Therefore, for any $k \ge 0$, the KL divergence converges exponentially fast:
    \begin{align}\label{Eq:KLProximalRateTotal}
        \KL(\rho_k^X \,\|\, \nu^X) \le \frac{\KL(\rho_0^X \,\|\, \nu^X)}{(1+\alpha \eta)^{2k}}.
    \end{align}
\end{lemma}

Note the exponential convergence rate in~\eqref{Eq:KLProximalRateTotal} matches the exponential convergence rate of the KL divergence along the Langevin dynamics~\eqref{Eq:LangevinKLRate}, since for small $\eta > 0$, $(1+\alpha \eta)^{-2k} \approx e^{-2\alpha \eta k}$.

%%%%%%%%%%%%%%%
\subsection{Convergence Rate of the Proximal Sampler in Relative Fisher Information}
\label{Sec:ProximalSampler}

We show that when the target distribution is strongly log-concave, relative Fisher information converges exponentially fast along the Proximal Sampler.
We use the same notation and interpretation of the forward and backward steps in Section~\ref{Sec:ProximalSamplerReviewKL}.
We provide the proof of Theorem~\ref{Thm:FIProximalRate} in Section~\ref{Sec:FIProximalRateProof}.

\begin{theorem}\label{Thm:FIProximalRate}
    Assume $\nu^X$ is $\alpha$-strongly log-concave for some $\alpha > 0$.
    Along each iteration of the Proximal Sampler with any $\eta > 0$, we have:
    \begin{subequations}\label{Eq:FIProximalRate}
        \begin{align}
            \FI(\rho_k^Y \,\|\, \nu^Y) \,&\le\, \frac{\FI(\rho_k^X \,\|\, \nu^X)}{(1+\alpha \eta)^2} \label{Eq:FIProximalRate-1} \\
            \FI(\rho_{k+1}^X \,\|\, \nu^X) \,&\le\, \FI(\rho_k^Y \,\|\, \nu^Y).  \label{Eq:FIProximalRate-2}       
        \end{align}        
    \end{subequations}
    Therefore, for any $k \ge 0$, the relative Fisher information converges exponentially fast:
    \begin{align}\label{Eq:FIProximalRateTotal}
        \FI(\rho_k^X \,\|\, \nu^X) \le \frac{\FI(\rho_0^X \,\|\, \nu^X)}{(1+\alpha \eta)^{2k}}.
    \end{align}
\end{theorem}

Note the exponential convergence rate in~\eqref{Eq:FIProximalRateTotal} matches the exponential convergence rate of the relative Fisher information along the Langevin dynamics~\eqref{Eq:LangevinFIRate}.
The convergence rate in~\eqref{Eq:FIProximalRateTotal} also matches the exponential convergence rate of the squared gradient norm along the proximal gradient method in optimization for a strongly convex objective function (see Lemma~\ref{Lem:OptReview}).
We can verify the convergence rate in~\eqref{Eq:FIProximalRateTotal} is tight in the Gaussian case; see Example~\ref{Ex:ConvergenceFIProx} in Section~\ref{Sec:ConvergenceFIProxExample}.

We note the convergence rate~\eqref{Eq:FIProximalRateTotal} in relative Fisher information is similar to the convergence rate~\eqref{Eq:KLProximalRateTotal} in KL divergence.
However, we also note the difference: For KL divergence, there is a contraction in both the forward and backward steps~\eqref{Eq:KLProximalRate}, since we have SDPI for both the forward and reverse Gaussian channels.
For relative Fisher information~\eqref{Eq:FIProximalRate}, there is a contraction only in the forward step since we have SDPI along the Gaussian channel (Theorem~\ref{Thm:FisherInfoHeatFlow}), while the backward step does not provide a strict contraction since we only have DPI along the reverse Gaussian channel.

%%%%%%%%%%%%%%%%%
\subsubsection{Convergence of Relative Fisher Information along Proximal Sampler}
\label{Sec:ConvergenceFIProxExample}

\begin{example}\label{Ex:ConvergenceFIProx}
    Consider the Proximal Sampler with step size $\eta > 0$ for the target distribution $\nu^X = \N(0, \alpha^{-1} I)$ on $\R^d$, which is $\alpha$-SLC, from the initial distribution $\rho_0^X = \N(m_0, \sigma_0^2 I)$ for some $m_0 \in \R^d$, $m_0 \neq 0$, and $\sigma_0^2 > 0$.
    In this case, the RGO conditional distribution is, for each $y \in \R^d$:
    $$\nu^{X \mid Y}(\cdot \mid y) = \N\left(\frac{1}{1+\alpha \eta} y, \frac{\eta}{1+\alpha \eta} I \right).$$
    Therefore, the iterates remain Gaussian $\rho_k^X = \N(m_k, \sigma^2_k I)$ for all $k \ge 0$, where:
    \begin{align*}
        m_{k} &= \frac{m_0}{(1+\alpha \eta)^k},
        \qquad\qquad
        \sigma_{k}^2 = \frac{1}{(1+\alpha \eta)^{2k}}\left(\sigma_0^2-\frac{1}{\alpha}\right) + \frac{1}{\alpha}.
    \end{align*}
    Then we can compute:
    \begin{align*}
        \FI(\rho_k^X \,\|\, \nu^X) 
        &\,=\, \alpha^2 \|m_k\|^2 + \left(\frac{1}{\sigma_k^2} - \alpha\right)^2 \sigma_k^2 d 
        \,=\, \frac{\alpha^2 \|m_0\|^2}{(1+\alpha \eta)^{2k}} + \frac{d}{\sigma_k^2} \frac{(1- \alpha \sigma_0^2)^2}{(1+\alpha \eta)^{4k}}.
    \end{align*}
    Since $m_0 \neq 0$, we have $\FI(\rho_k^X \,\|\, \nu^X) = O((1+\alpha \eta)^{-2k})$, as claimed in Theorem~\ref{Thm:FIProximalRate}.
    (Note if $m_0 = 0$, then we have an improved rate $\FI(\rho_k^X \,\|\, \nu^X) = O((1+\alpha \eta)^{-4k})$; this is because in this case $\rho_k^X$ is symmetric and satisfies a Poincar\'e inequality, so the improved SDPI from Theorem~\ref{Thm:FisherInfoHeatFlow}(iii) applies.)
\end{example}

%%%%%%%%%%%%%%%%%%%%%%%%
\subsection{Iteration Complexity of the Proximal Sampler in Relative Fisher Information}

With the rejection sampling implementation of RGO (see Section~\ref{Sec:RejectionReview}), the convergence rate in Theorem~\ref{Thm:FIProximalRate} implies the following high-accuracy iteration complexity guarantee for the Proximal Sampler in relative Fisher information.
We provide the proof of Corollary~\ref{Cor:IterationComplexity} in Section~\ref{Sec:IterationComplexityProof}.

\begin{corollary}\label{Cor:IterationComplexity}
    Assume $\nu^X \propto e^{-g}$ is $\alpha$-strongly log-concave and $L$-log-smooth for some $0 < \alpha \le L < \infty$.
    Let $x^* = \arg\min_{x \in \R^d} g(x)$ be the minimizer of $g$.
    Consider the Proximal Sampler with step size $\eta = \frac{1}{Ld}$ with the rejection sampling implementation of the RGO.
    Suppose we start the algorithm from $x_0 \sim \rho_0^X = \N(x^*, \frac{1}{L} I)$.
    Then for any $\varepsilon > 0$, along the iteration of the Proximal Sampler $x_k \sim \rho_k^X$, we have $\FI(\rho_k^X \,\|\, \nu^X) \le \varepsilon$ whenever the number of iterations $k$ satisfies
    $$k \ge \frac{dL}{\alpha} \log \frac{dL}{\varepsilon}$$
    and the expected number of queries to $g$ in each iteration for the rejection sampling is $O(1)$.
\end{corollary}

In the above, we assume we can find the minimizer $x^*$ of $g$, for example by gradient descent, which converges exponentially fast since we assume $g$ is smooth and strongly convex.
We note the logarithmic dependence of the iteration complexity above on the error parameter $\varepsilon^{-1}$, which comes from the exponential convergence rate of the Proximal Sampler in relative Fisher information.
This is in contrast to the polynomial dependence on $\varepsilon^{-1}$ in the result of~\citep{balasubramanian22a}, which comes from the bias of ULA.
However, the result of~\citep{balasubramanian22a} only assumes the target is log-smooth, while our result assumes the target is SLC and log-smooth.
Furthermore, our guarantee holds for the last iterate of the Proximal Sampler, while the guarantee in~\citep{balasubramanian22a} holds for the average iterate of ULA.

%%%%%%%%%%%%%
\section{Discussion}
\label{Sec:Discussion}

We have shown that the Proximal Sampler has a high-accuracy iteration complexity guarantee for sampling in relative Fisher information when the target distribution is strongly log-concave and log-smooth.
This complements the result of~\citep{balasubramanian22a}, who show a low-accuracy iteration complexity guarantee of ULA for non-log-concave, log-smooth target distribution.

Our analysis technique is to establish SDPI in relative Fisher information along several channels, including the Gaussian channel under strong log-concavity, and along the reverse Gaussian channel, which compose to form the Proximal Sampler algorithm.
We show that DPI in relative Fisher information may not hold initially along the Gaussian channel and the OU process.
We also show that SDPI in relative Fisher information will eventually hold along the Gaussian channel when the second input distribution is a log-Lipschitz perturbation of a strongly log-concave distribution, and along the OU process when the second input distribution is strongly log-concave.

Going forward, it would be interesting to study if we can weaken the SLC assumption to isoperimetry such as LSI.
It would be interesting to show mixing time in relative Fisher information for other sampling algorithms beyond what we study in this paper, including for discretization of the underdamped Langevin dynamics~\citep{ma2021there} and for the Hamiltonian Monte Carlo algorithm~\citep{monmarche2024entropic}.

%%%%%%%%%%%%
\appendix

%%%%%%%%%%%%%%%%%%%%%
\section{Additional Related Works}
\label{Sec:RelatedWorks}

Mixing time guarantees in relative Fisher information have been studied in~\citep{balasubramanian22a,chewi2023fisher}, both focusing on non-log-concave sampling as we discussed in Section~\ref{Sec:Introduction}.
The work of~\citep{cheng2023fast} uses convergence in relative Fisher information to characterize conditional mixing for non-log-concave target distributions.

For other sampling algorithms, mixing time analyses based on hypocoercivity~\citep{villani2009hypocoercivity} involve the construction of a Lyapunov function which combines the KL divergence and a rescaled relative Fisher information; see~\citep{ma2021there} for the analysis of the underdamped Langevin algorithm, and~\citep{monmarche2024entropic} for the analysis of the ideal Hamiltonian Monte Carlo algorithm, where both work only assume the target distribution satisfy LSI and smoothness.
From such Lyapunov analyses, one can in principle extract convergence guarantees in relative Fisher information for the corresponding algorithms.
However, the result will be a biased convergence guarantee for the underdamped Langevin algorithm in~\citep{ma2021there} since the discretization introduces bias similar to ULA.
On the other hand, the result will be an exponential convergence guarantee for the ideal Hamiltonian Monte Carlo algorithm in~\citep{monmarche2024entropic} since the algorithm is unbiased, but the algorithm is still idealized and requires a further discretization step to implement in practice, and it is unclear whether the convergence guarantee in relative Fisher information will be preserved.

Convergence guarantees in relative Fisher information have been useful as an intermediate step in showing the uniqueness of the equilibrium distribution and analyzing the convergence behaviors of the mean-field Langevin dynamics for zero-sum games in the space of probability distributions under strong convexity, see e.g.~\cite[Lemma~7.2]{conger2024coupled} and~\cite[Theorem~4]{cai2024convergence}.

The work of~\citep{feng2024fisher} study the decay of Fisher information along a general class of stochastic processes which allow a time and position-dependent covariance term multiplying the (possibly degenerate) Brownian motion; this includes Fokker-Planck channels as a special case.
\citep{feng2024fisher} derive a formula for the time derivative of a generalized relative Fisher information along such processes,
using the machinery of information gamma calculus, and extract convergence guarantees in relative Fisher information under certain assumptions, see~\cite[Theorem~2]{feng2024fisher}; they provide various applications including analysis of annealed Langevin dynamics.
In this work, we provide a self-contained proof of the time derivative formula along the Fokker-Planck channel (see Lemma~\ref{Lem:TimeDerivativeFI}); 
furthermore, our SDPI perspective and our algorithmic application for deriving the mixing time of the Proximal Sampler in relative Fisher information seem novel.

\subsection{Parametric vs.\ Non-Parametric Fisher Information}
\label{Sec:RelatedWorksFI}

In this paper, we are working with the \textit{non-parametric} relative Fisher information $\FI(\rho\,\|\,\nu)$, where the derivative of the log-density is taken with respect to the space variable; this has a natural geometric interpretation as the squared gradient norm of the KL divergence in the Wasserstein metric.
The relative Fisher information is a relative form of the non-parametric Fisher information $J(\rho) = \E_\rho[\|\nabla \log \rho\|^2]$. 
More precisely, the relative Fisher information $\FI(\rho\,\|\,\nu)$ is the Bregman divergence (i.e., excess in the linear approximation) of the absolute Fisher information; see e.g.~\cite[Lemma~2]{wibisono2017information}.

In the literature, another popular notion is the {\em parametric} Fisher information, where 
the derivative of the log-density is taken with respect to a parameter that indexes the distribution; for example, this is the notion of Fisher information that appears in the Cramer-Rao bound.
When the parameter is a location parameter, i.e., the distribution is $\rho_\theta(x) = \rho_0(x-\theta)$ for some base distribution $\rho_0$ and parameter $\theta$, then the parametric Fisher information coincides with the non-parametric Fisher information; however, in general they are different.
In particular, the non-parametric Fisher information requires working with a continuous state space such as $\R^d$, while the parametric Fisher information is well-defined even for a discrete state space as long as the parameter is a continuous variable.

The data processing inequality for the parametric Fisher information is known under post-processing, see e.g.~\cite[Lemma~3]{zamir1998proof}.
For the non-parametric Fisher information $J(\rho)$, we have the \textit{Fisher information inequality} or the Blachman-Stam inequality, which states that for any probability distributions $\rho, \nu \in \P(\R^d)$: 
$$J(\rho \ast \nu)^{-1} \ge J(\rho)^{-1} + J(\nu)^{-1}$$
see e.g.~\citep{blachman1965convolution} and \cite[Theorem~1]{zamir1998proof}.
In particular, for any distribution $\rho$, we have $J(\rho \ast \nu) \le J(\rho)$, which can be seen as a data processing inequality in non-parametric Fisher information for a noisy channel which performs convolution with $\nu$ (i.e., adding an independent noise drawn from $\nu$).
However, a version of the inequality above for relative Fisher information is not known, and does not appear to be true in general; in Example~\ref{Ex:CounterDecrHeat}, we show an example where convolution with a Gaussian (i.e., along the Gaussian channel) increases relative Fisher information.

%%%%%%%%%%%%%%%%%
\section{Review of Sampling as Optimization in the Space of Distributions}
\label{Sec:SamplingAsOpt}

In the perspective of sampling as optimization in the space of distributions~\citep{JKO98,W18}, the problem of sampling from a target probability distribution $\nu \in \P(\R^d)$ can be formulated as minimizing some objective functions $\F \colon \P(\R^d) \to \R$ on the space of probability distributions which is minimized at the target distribution $\rho^* = \nu$; for example, a good objective function is the KL divergence $\F(\rho) = \KL(\rho \,\|\, \nu)$.

We recall that the Langevin dynamics~\eqref{Eq:Langevin} to the target distribution $\nu \propto e^{-g}$, or more precisely its associated Fokker-Planck equation~\eqref{Eq:FPLangevin}, has an optimization interpretation as the gradient flow dynamics: 
$$\dot \rho_t = -\grad \, \KL(\rho_t \,\|\, \nu)$$ 
for minimizing the KL divergence objective function $\F(\rho) = \KL(\rho \,\|\, \nu)$ on the space of probability distributions $\P(\R^d)$ endowed with the Wasserstein $W_2$ metric~\citep{JKO98}. 
Here $\grad$ is the gradient operator with respect to the $W_2$ metric, following the computation rule of Otto calculus, see e.g.~\cite[Ch.~15]{villani2009optimal}.

We also recall the optimization interpretation of the definitions above~\citep{OV00}: if $\nu$ is $\alpha$-SLC, then the KL divergence $\rho \mapsto \KL(\rho\,\|\,\nu)$ is an $\alpha$-geodesically strongly convex function on the space of probability distributions on $\R^d$ with respect to the $W_2$ metric.
Furthermore, recall the relative Fisher information is the squared gradient norm of the KL divergence in the $W_2$ metric: 
$$\FI(\rho\,\|\,\nu) = \big\|\grad \, \KL(\rho\,\|\,\nu) \big\|^2_\rho.$$
The property that $\nu$ satisfies LSI is equivalent to the Polyak-\L{}ojaciewicz (PL) or gradient domination property of the KL divergence objective function, which means for all probability distributions $\rho$:
$$\big \|\grad\,\KL(\rho\,\|\,\nu) \big\|^2_\rho \ge 2\alpha \, \KL(\rho\,\|\,\nu).$$
We recall that $\alpha$-strong convexity implies $\alpha$-gradient domination, and gradient domination is a sufficient condition for the exponential convergence rate in function value along gradient flow.

Via the gradient flow interpretation of the Langevin dynamics, de Bruijn's identity~\eqref{Eq:deBruijn} can be seen as the energy identity along gradient flow: $$\frac{d}{dt} \KL(\rho_t\,\|\,\nu) = - \big\|\grad \, \KL(\rho_t\,\|\,\nu) \big\|^2_{\rho_t}.$$
Thus, the exponential convergence rate of KL divergence along the Langevin dynamics under LSI~\eqref{Eq:LangevinKLRate} is a manifestation of the fact that along the gradient flow for an objective function which is gradient dominated, the function value converges exponentially fast to its minimum.

Via the gradient flow interpretation of the Langevin dynamics ($\dot \rho_t = -\grad \, \KL(\rho_t \,\|\, \nu)$), 
the identity~\eqref{Eq:LangevinFI} also follows from a general time derivative computation along the gradient flow:
$$\frac{d}{dt} \left\|\grad \, \KL(\rho_t\,\|\,\nu) \right\|^2_{\rho_t} = -2(\Hess \, \KL(\rho_t\,\|\,\nu))\left[\left(\grad \, \KL(\rho_t \,\|\, \nu) \right)^{\otimes 2} \right]$$ where $\Hess$ is the Hessian operator in the $W_2$ metric, which when evaluated yields the right-hand side of the identity~\eqref{Eq:LangevinFI} above; see~\cite[Ch.~15]{villani2009optimal}.
Thus, the exponential convergence rate~\eqref{Eq:LangevinFIRate} of the relative Fisher information along the Langevin dynamics under SLC is a manifestation of the fact that along the gradient flow for a strongly convex objective function, the squared gradient norm converges exponentially fast; see Lemma~\ref{Lem:OptReview} in Section~\ref{Sec:ReviewOpt} for a review.

%%%%%%%%%%%%%%%%%
\section{Review of the Gradient Norm Convergence in Optimization}
\label{Sec:ReviewOpt}

We review the behavior of the squared gradient norm in finite-dimensional optimization.

Suppose we want to minimize an objective function $f \colon \R^d \to \R$.
We assume $f$ is twice differentiable and $\alpha$-strongly convex,
so the Hessian matrix satisfies $\nabla^2 f(x) \succeq \alpha I$ for all $x \in \R^d$.

In continuous time, we consider the \textit{gradient flow} dynamics for $(X_t)_{t \ge 0}$ in $\R^d$:
\begin{align}\label{Eq:GF}
    \dot X_t = -\nabla f(X_t).
\end{align}

In discrete time, we consider the \textit{proximal gradient algorithm} for the iterates $x_k \in \R^d$:
\begin{align}\label{Eq:PGM}
    x_{k+1} = \arg\min_{x \in \R^d} \left\{ f(x) + \frac{1}{2\eta} \|x-x_k\|^2 \right\}
\end{align}
where $\eta > 0$ is step size.
Recall that the proximal gradient method algorithm can be equivalently written as the following implicit update:
\begin{align*}
    x_{k+1} = x_k - \eta \nabla f(x_{k+1}).
\end{align*}

We have the following guarantees for the convergence of the squared gradient norm along the gradient flow and the proximal gradient algorithm above.
Note the rate for the proximal gradient algorithm below matches the convergence rate for the gradient flow: $(1+\alpha \eta)^{2k} \to \exp(-2\alpha t)$, in the limit of vanishing step size $\eta \to 0$ with $\eta k \to t$.

\begin{lemma}\label{Lem:OptReview}
  Assume $f$ is $\alpha$-strongly convex and twice continuously differentiable.
  \begin{enumerate}
  \item Along the gradient flow dynamics~\eqref{Eq:GF}, we have for all $t \ge 0$:
  \begin{align*}
	\|\nabla f(X_t)\|^2 \le e^{-2\alpha t} \|\nabla f(X_0)\|^2.
  \end{align*}
  \item Along the proximal gradient algorithm~\eqref{Eq:PGM} with any step size $\eta > 0$, we have for all $k \ge 0$:
  \begin{align*}
	\|\nabla f(x_k)\|^2 \le \frac{\|\nabla f(x_0)\|^2}{(1+\alpha \eta)^{2k}}.
  \end{align*}
  \end{enumerate}
\end{lemma}
\begin{proof}
\textbf{(1) Analysis for gradient flow:}
Using the gradient flow dynamics~\eqref{Eq:GF}, we can compute:
\begin{align*}
    \frac{d}{dt} \|\nabla f(X_t)\|^2 &= 2\langle \nabla f(X_t), \nabla^2 f(X_t) \, \dot X_t \rangle 
    = -2\|\nabla f(X_t)\|^2_{\nabla^2 f(X_t)} 
    \le -2\alpha \|\nabla f(X_t)\|^2,
\end{align*}
where in the last step above we have used the property that $\nabla^2 f(X_t) \succeq \alpha I$.
We can write the differential inequality above equivalently as: $\frac{d}{dt} \left(e^{2\alpha t}  \|\nabla f(X_t)\|^2\right) \le 0$.
Integrating from $0$ to $t$ and rearranging, we get the desired convergence bound along gradient flow:
\begin{align*}
    \|\nabla f(X_t)\|^2 \le e^{-2\alpha t} \|\nabla f(X_0)\|^2.
\end{align*}

\medskip
\noindent
\textbf{(2) Analysis for proximal gradient algorithm:}
We use a continuous-time interpolation of each step of the proximal gradient algorithm as follows.
From any $Y_0 \in \R^d$, consider the continuous-time curve $(Y_t)_{t \ge 0}$ defined by:
\begin{align*}
	Y_t = Y_0 - t \nabla f(Y_t).
\end{align*}
Observe that if $Y_0 = x_k$, then the next iterate $x_{k+1}$ of the proximal gradient algorithm is equal to the value $Y_\eta$ at time $t = \eta$, since indeed at $t = \eta$, we have $Y_\eta = Y_0 - \eta \nabla f(Y_\eta) = x_k - \eta \nabla f(Y_\eta)$, so $Y_\eta = (I + \eta \nabla f)^{-1}(x_k) = x_{k+1}$.

Along the curve $Y_t$ above, we can compute:
$\dot Y_t = -\nabla f(Y_t) - t \nabla^2 f(Y_t) \dot Y_t$,
and therefore,
\begin{align*}
	\dot Y_t = -(I + t \nabla^2 f(Y_t))^{-1} \nabla f(Y_t).
\end{align*}
Therefore, we can compute along the curve $Y_t$ above:
\begin{align*}
	\frac{d}{dt} \|\nabla f(Y_t)\|^2 &= 2\langle \nabla f(Y_t), \, \nabla^2 f(Y_t) \, \dot Y_t \rangle 
	= -2\|\nabla f(Y_t)\|^2_{\nabla^2 f(Y_t) (I + t \nabla^2 f(Y_t))^{-1}}.
\end{align*}
Note since $\nabla^2 f(Y_t) \succeq \alpha I$, we have $\nabla^2 f(Y_t) (I + t \nabla^2 f(Y_t))^{-1} \succeq \frac{\alpha}{1+\alpha t} I$ for all $t \ge 0$.
Therefore, continuing the calculation above,
\begin{align*}
	\frac{d}{dt} \|\nabla f(Y_t)\|^2 &\le -\frac{2\alpha}{1+\alpha t}  \|\nabla f(Y_t)\|^2 = -\dot A_t  \|\nabla f(Y_t)\|^2
\end{align*}
where in the last step we have defined $A_t = 2 \log (1+\alpha t)$, so $\dot A_t = \frac{2\alpha}{1+\alpha t}$.
The above differential inequality is equivalent to: $\frac{d}{dt} \left(e^{A_t}  \|\nabla f(Y_t)\|^2  \right) \le 0$.
Integrating from $0$ to $t$ and rearranging yields:
\begin{align*}
	 \|\nabla f(Y_t)\|^2 \le e^{-A_t}  \|\nabla f(Y_0)\|^2 = \frac{ \|\nabla f(Y_0)\|^2 }{(1+\alpha t)^2}. 
\end{align*}
Plugging in $t = \eta$, and substituting $Y_0 = x_k$, $Y_\eta = x_{k+1}$, we obtain the one-step contraction bound for the proximal gradient algorithm:
\begin{align*}
	 \|\nabla f(x_{k+1})\|^2 \le \frac{ \|\nabla f(x_k)\|^2}{(1+\alpha \eta)^2}.
\end{align*}
Unrolling this recursion gives the desired convergence bound for the proximal gradient algorithm.
\end{proof}

%%%%%%%%%%%%%%%%%%%%%%%%%%%%
\section{Large Gap between Relative Fisher Information and KL Divergence}
\label{Sec:LargeFIEx}

We exhibit an example of two distributions where their relative Fisher information is large, but the KL divergence is small. 
In fact, we can make the R\'enyi divergence of order $\infty$ be small.

Recall the \textit{R\'enyi divergence} of order $1 < q < \infty$ between probability distributions $\rho$ and $\nu$ is:
$$\sR_q(\rho\,\|\,\nu) = \frac{1}{q-1} \log \E_\nu\left[\left(\frac{\rho}{\nu}\right)^q\right].$$
Recall that $q \mapsto \sR_q(\rho\,\|\,\nu)$ is increasing, $\lim_{q \to 1} \sR_q(\rho\,\|\,\nu) = \KL(\rho\,\|\,\nu)$, and the R\'enyi divergence of order $\infty$ is the logarithm of the \textit{warmness} parameter:
$$\sR_\infty(\rho\,\|\,\nu) = \lim_{q \to \infty} \sR_q(\rho\,\|\,\nu) = \ess\sup_\nu \log \frac{\rho}{\nu}$$
where $\ess\sup_\nu$ denotes the supremum over the support of $\nu$ except on sets of measure zero with respect to $\nu$.
Therefore, if we can show $\sR_\infty(\rho\,\|\,\nu)$ is small, then it implies $\sR_q(\rho\,\|\,\nu)$ is small for all $1 < q < \infty$, and in particular, $\KL(\rho\,\|\,\nu)$ is also small.

We have the following construction in $d=1$ dimension, where the second distribution $\nu$ is a standard Gaussian (which is SLC and thus satisfies LSI), and the first distribution $\rho$ is a small perturbation of $\nu$, such that the $\infty$-R\'enyi divergence can be made arbitrarily small, but the relative Fisher information can be made arbitrarily large.

\begin{lemma}\label{Lem:LargeFI}
    Let $\nu = \N(0,1)$ on $\R$.
    Let $0 < \epsilon < 1 < L$ be arbitrary.
    There exists a probability distribution $\rho$ on $\R$ (constructed below) such that $\sR_\infty(\rho \,\|\, \nu) \le \epsilon$ and $\FI(\rho \,\|\, \nu) \ge L$.
\end{lemma}
\begin{proof}
    Define $a \equiv a(\epsilon) > 0$ such that $\nu([-a,a]) = \epsilon$.
    Define $M = \max\{1/a, \sqrt{(e L)/\epsilon}\}$.
    Let $K \ge 0$ be the smallest integer such that $(2K+1)/M \ge a$.
    Define $\eta := a/(2K+1)$, so $\eta \le \frac{1}{M} \le a$.
        
    We define $g \colon \R \to \R$ to be a periodic spike function with height $1$ and width $2\eta$ supported on the interval $[-a,a]$, that is:
    \begin{align*}
        g(x) &= \begin{cases}
            1-|x - 2k \eta|/\eta ~ & \text{ if } |x| \le a \text{ and } x \in [(2k-1)\eta, \, (2k+1)\eta] \text{ for } k \in \{-K,\dots,K\}, \\
            0 & \text{ if } |x| > a.
        \end{cases}
    \end{align*}
    
    Define the probability distribution $\rho \propto \nu \cdot e^{-g}$, or explicitly:
    $$\rho(x) = \frac{\nu(x) \, e^{-g(x)}}{\E_\nu[e^{-g}]}.$$

    Since $0 \le g(x) \le 1$ for $|x| \le a$ and $g(x) = 0$ for $|x| > a$, we have 
    $\E_\nu[g] \le \nu([-a,a]) = \epsilon.$
    Then by Jensen's inequality,
    $\log \E_\nu[e^{-g}] \ge -\E_\nu[g] \ge -\epsilon.$
    Then for all $x \in \R^d$, we have:
    $$\log \frac{\rho(x)}{\nu(x)} = -g(x) - \log \E_\nu[e^{-g}] \le 0 + \epsilon = \epsilon.$$
    Therefore, $\sR_\infty(\rho\,\|\,\nu) \le \epsilon$, as claimed.

    Next, since $g(x) \ge 0$, we have $\E_\nu[e^{-g}] \le 1$.
    Since $|g'(x)| = 1/\eta$ almost everywhere for $|x| \le a$ (except for $x$ which are integer multiples of $\eta$) and $g'(x) = 0$ for $|x| > a$, and since $g(x) \le 1$, we have
    $$\E_\nu[(g')^2 e^{-g}] = \frac{1}{\eta^2} \int_{-a}^a e^{-g(x)} \nu(x) \, dx
    \ge \frac{e^{-1}}{\eta^2} \nu([-a,a]) = \frac{e^{-1} \epsilon}{\eta^2} \ge e^{-1} \epsilon M^2 \ge L.$$
    Therefore, we can bound the relative Fisher information as
    \begin{align*}
        \FI(\rho\,\|\,\nu) &= \E_\rho\left[\left(\log \frac{\rho}{\nu}\right)^{'2} \right] 
        = \E_\rho[(g')^2]
        = \frac{\E_\nu[(g')^2 e^{-g}]}{\E_\nu[e^{-g}]}
        \ge L
    \end{align*}
    as claimed.
\end{proof}

%%%%%%%%%%%%%%%%%%%%%%
\section{Proof of Lemma~\ref{Lem:TimeDerivativeKL} (SDPI for KL Divergence along the FP Channel)}
\label{Sec:TimeDerivativeKLProof}

We provide the proof below assuming the distributions have sufficient regularity properties for the various operations to hold.
See also~\cite{klartag25heat} for a rigorous derivation of the identity in the case of the heat flow.

\begin{proof}[Proof of Lemma~\ref{Lem:TimeDerivativeKL}]
    We can compute by differentiating under the integral sign and using chain rule:
    \begin{align}
        \frac{d}{dt} \KL(\rho_t \,\|\, \nu_t)
        &= \frac{d}{dt} \int_{\R^d} \rho_t \log \frac{\rho_t}{\nu_t} \, dx \notag \\
        &= \int_{\R^d} (\partial_t \rho_t) \log \frac{\rho_t}{\nu_t} \, dx + \int_{\R^d} \rho_t \frac{1}{\rho_t} \partial_t \rho_t \, dx - \int_{\R^d} \rho_t \frac{1}{\nu_t} \partial_t \nu_t \, dx. \label{Eq:ddtCalc1}
    \end{align}
    We compute each term above.
    For the first term in~\eqref{Eq:ddtCalc1}, 
    by the Fokker-Planck equation and  using integration by parts, where all the boundary terms vanish:
    \begin{align*}
        \int_{\R^d} (\partial_t \rho_t) \log \frac{\rho_t}{\nu_t} \, dx
        &= \int_{\R^d} \left(-\nabla \cdot (\rho_t b_t) + \frac{c}{2} \Delta \rho_t\right) \log \frac{\rho_t}{\nu_t} \, dx \\
        &= \int_{\R^d} \rho_t \left\langle b_t, \nabla \log \frac{\rho_t}{\nu_t} \right\rangle \, dx
        - \frac{c}{2} \int_{\R^d} \rho_t \left\langle \nabla \log \rho_t, \nabla \log \frac{\rho_t}{\nu_t} \right\rangle \, dx
    \end{align*}
    where for the second term above we have used the identity that $\Delta \rho_t = \nabla \cdot \nabla \rho_t = \nabla \cdot (\rho_t \nabla \log \rho_t)$.

    For the second term in~\eqref{Eq:ddtCalc1}, we can show it is equal to $0$:
    \begin{align*}
        \int_{\R^d} \rho_t \frac{1}{\rho_t} \partial_t \rho_t \, dx
        &= \int_{\R^d} \partial_t \rho_t \, dx
        \,=\, \partial_t \left(\int_{\R^d} \rho_t \, dx \right)
        \,=\, \partial_t (1)
        \,=\, 0.
    \end{align*}

    For the third term in~\eqref{Eq:ddtCalc1}, 
    by the Fokker-Planck equation and  using integration by parts:
    \begin{align*}
        - \int_{\R^d} \rho_t \frac{1}{\nu_t} \partial_t \nu_t \, dx
        &= -\int_{\R^d} \left(-\nabla \cdot (\nu_t b_t) + \frac{c}{2}  \Delta \nu_t\right) \frac{\rho_t}{\nu_t} \, dx \\
        &= - \int_{\R^d} \nu_t \left\langle b_t, \nabla \frac{\rho_t}{\nu_t} \right\rangle \, dx
        + \frac{c}{2}  \int_{\R^d} \nu_t \left\langle \nabla \log \nu_t, \nabla \frac{\rho_t}{\nu_t} \right\rangle \, dx \\
        &= - \int_{\R^d} \rho_t \left\langle b_t, \nabla \log \frac{\rho_t}{\nu_t} \right\rangle \, dx
        + \frac{c}{2}  \int_{\R^d} \rho_t \left\langle \nabla \log \nu_t, \nabla \log \frac{\rho_t}{\nu_t} \right\rangle \, dx
    \end{align*}
    where in the above we have used the identity that $\Delta \nu_t = \nabla \cdot (\nu_t \nabla \log \nu_t)$,
    and $\nu_t \nabla \frac{\rho_t}{\nu_t} = \rho_t \nabla \log \frac{\rho_t}{\nu_t}$.

    Combining the calculations above in~\eqref{Eq:ddtCalc1}, we see that the terms involving $\left\langle b_t, \nabla \log \frac{\rho_t}{\nu_t} \right\rangle$ cancel, and we obtain:
    \begin{align*}
        \frac{d}{dt} \KL(\rho_t \,\|\, \nu_t)
        &= 
        - \frac{c}{2} \int_{\R^d} \rho_t \left\langle \nabla \log \rho_t, \nabla \log \frac{\rho_t}{\nu_t} \right\rangle \, dx 
        + \frac{c}{2}  \int_{\R^d} \rho_t \left\langle \nabla \log \nu_t, \nabla \log \frac{\rho_t}{\nu_t} \right\rangle \, dx \\
        &= 
        - \frac{c}{2} \int_{\R^d} \rho_t \left\| \nabla \log \frac{\rho_t}{\nu_t} \right\|^2 dx  \\
        &= 
        - \frac{c}{2}  \, \FI(\rho_t \,\|\, \nu_t)
    \end{align*}
    as claimed in~\eqref{Eq:ddtKLGeneral}.

    Now suppose for all $t \ge 0$ we know that $\nu_t$ satisfies $\alpha_t$-LSI.
    Then from the identity above, we can bound:
    $$\frac{d}{dt} \KL(\rho_t \,\|\, \nu_t) = - \frac{c}{2}  \, \FI(\rho_t \,\|\, \nu_t) \le -c \alpha_t \, \KL(\rho_t \,\|\, \nu_t).$$
    Integrating this differential inequality yields
    \begin{align*}
        \KL(\rho_t \,\|\, \nu_t) \le \exp\left(-c \int_0^t \alpha_s \, ds \right) \KL(\rho_0 \,\|\, \nu_0)
    \end{align*}
    as claimed in~\eqref{Eq:SDPI_KL}.
\end{proof}

%%%%%%%%%%%%%%%%%%%%%%%
\section{Proofs for Section~\ref{Sec:ContractionFI} (Contraction of Relative Fisher Information)}

%%%%%%%%%%%%%%%%%%
\subsection{Proof of Lemma~\ref{Lem:TimeDerivativeFI} (Time Derivative of Fisher Information along FP Channel)}
\label{Sec:TimeDerivativeFIProof}

\begin{replemma}{Lem:TimeDerivativeFI}
    Suppose $(\rho_t)_{t \ge 0}$ and $(\nu_t)_{t \ge 0}$ evolve following the Fokker-Planck equation~\eqref{Eq:FP}:
    \begin{subequations}\label{Eq:FPGeneral}
        \begin{align}
            \partial_t \rho_t &= -\nabla \cdot (\rho_t b_t) + \frac{c}{2} \Delta \rho_t \enspace, \label{Eq:FPrhot} \\
            \partial_t \nu_t &= -\nabla \cdot (\nu_t b_t) + \frac{c}{2} \Delta \nu_t \enspace.  \label{Eq:FPnut}
        \end{align}
    \end{subequations}
    Then for any $t \ge 0$:
    \begin{align*}
        \frac{d}{dt} \FI(\rho_t \,\|\, \nu_t) 
        = -c \, \E_{\rho_t} \left[\left\| \nabla^2 \log \frac{\rho_t}{\nu_t} \right\|^2_{\HS} \right] 
        - 2\E_{\rho_t}\left[\left\|\nabla \log \frac{\rho_t}{\nu_t} \right\|^2_{(-c \, \nabla^2 \log \nu_t + (\nabla b_t)_{\sym})}\right].
    \end{align*}
\end{replemma}
\begin{proof}
    Let $f_t = -\log \rho_t$ be the negative log-density of $\rho_t$, and note the relation:
    \begin{align*}
        \Delta \rho_t &= - \rho_t \Delta f_t + \rho_t \|\nabla f_t\|^2.
    \end{align*}
    Since $\rho_t$ evolves following the Fokker-Planck equation~\eqref{Eq:FPrhot}, we can compute the evolution for $f_t$:
    \begin{align}
        \partial_t f_t = -\frac{1}{\rho_t} \partial_t \rho_t 
        &= -\frac{1}{\rho_t} \left( -\langle \nabla \rho_t, b_t \rangle - \rho_t \nabla \cdot b_t + \frac{c}{2} \Delta \rho_t \right) \notag \\
        &= -\langle \nabla f_t, b_t \rangle + \nabla \cdot b_t + \frac{c}{2} \Delta f_t - 
        \frac{c}{2} \|\nabla f_t\|^2. \label{Eq:ft}
    \end{align}
    Similarly, let $g_t = -\log \nu_t$. 
    Since $\nu_t$ evolves following~\eqref{Eq:FPnut}, we can compute the evolution for $g_t$:
    \begin{align}\label{Eq:gt}
        \partial_t g_t = -\langle \nabla g_t, b_t \rangle + \nabla \cdot b_t + \frac{c}{2} \Delta g_t - 
        \frac{c}{2} \|\nabla g_t\|^2.
    \end{align}

    In the computation below, we will use the following relations:
    \begin{subequations}
        \begin{align}
            \rho_t \nabla f_t &= -\nabla \rho_t \label{Eq:Rel1} \\
            \rho_t \nabla^2 f_t &= -\nabla^2 \rho_t + \rho_t (\nabla f_t) (\nabla f_t)^\top. \label{Eq:Rel2}
        \end{align}
    \end{subequations}
    We also use the \textit{Bochner's formula}, which states that for smooth $u \colon \R^d \to \R$,
    \begin{align}\label{Eq:Bochner}
        \langle \nabla u, \nabla \Delta u \rangle = \frac{1}{2} \Delta \|\nabla u\|^2 - \|\nabla^2 u\|^2_{\HS}.
    \end{align}

    Using relation~\eqref{Eq:Rel1} and integration by parts, we can write the relative Fisher information as:
    \begin{align}
        \FI(\rho_t \,\|\, \nu_t)
        \,&=\, \E_{\rho_t}\left[\|\nabla f_t - \nabla g_t\|^2 \right] \notag \\
        &= \E_{\rho_t}\left[\|\nabla f_t\|^2 + \|\nabla g_t\|^2 - 2 \langle \nabla f_t, \nabla g_t \rangle \right] \notag \\
        &= \blue{\E_{\rho_t}\left[\|\nabla f_t\|^2 \right]} +  
        \red{\E_{\rho_t}\left[\|\nabla g_t\|^2\right]} 
        \purple{- 2 \E_{\rho_t}\left[\Delta g_t \right]}. \label{Eq:FIt}
    \end{align}
    We compute the time derivative of each term above, spelling out the computations below explicitly, and color coding them for clarity.
    In the computations below, we assume the density functions $\rho_t, \nu_t$ are smooth and decay sufficiently fast at infinity.
    Hence, we can differentiate under the integral sign; furthermore, when we perform integration by parts, all the boundary terms vanish.
    
    %%%%%%%%%%%%
    \paragraph{(I) First term:} 
    The first term in~\eqref{Eq:FIt} is the Fisher information of $\rho_t$.    
    We compute its time derivative using the Fokker-Planck equation~\eqref{Eq:FPrhot} and the formula~\eqref{Eq:ft} to get:
    \begin{subequations}
    \begin{align}
        \frac{d}{dt} \E_{\rho_t}\left[\|\nabla f_t\|^2 \right] 
        &= \int_{\R^d} (\partial_t \rho_t) \,  \|\nabla f_t\|^2 \, dx + 2\int_{\R^d} \rho_t \, \langle \nabla f_t, \nabla \partial_t f_t \rangle \, dx \notag \\
        &= -\int_{\R^d} \nabla \cdot (\rho_t b_t) \,  \|\nabla f_t\|^2 \, dx + \frac{c}{2} \int_{\R^d} (\Delta \rho_t) \,  \|\nabla f_t\|^2 \, dx  \label{Eq:Calc1a}  \\
        &\qquad
        + 2 \int_{\R^d} \rho_t \, \langle \nabla f_t, \nabla ( -\langle \nabla f_t, b_t \rangle + \nabla \cdot b_t) \rangle \, dx \label{Eq:Calc1b} \\
        &\qquad 
        + c \int_{\R^d} \rho_t \, \langle \nabla f_t, \nabla (\Delta f_t) \rangle \, dx 
        - c \int_{\R^d} \rho_t \, \langle \nabla f_t, \nabla \|\nabla f_t\|^2 \rangle \, dx. \label{Eq:Calc1c} 
    \end{align}
    \end{subequations}
    We calculate the terms above one by one.
    \begin{enumerate}
        \item The first term in~\eqref{Eq:Calc1a} is, by integration by parts:
        \begin{align*}
            -\int_{\R^d} \nabla \cdot (\rho_t b_t) \,  \|\nabla f_t\|^2 \, dx
            &= \int_{\R^d} \rho_t \langle b_t, \, \nabla (\|\nabla f_t\|^2) \rangle \, dx 
            = 2\int_{\R^d} \rho_t \langle b_t, \, (\nabla^2 f_t) \, \nabla f_t \rangle \, dx.            
        \end{align*}
        \item The second term in~\eqref{Eq:Calc1a} is, by integration by parts:
        \begin{align*}
            \frac{c}{2} \int_{\R^d} (\Delta \rho_t) \,  \|\nabla f_t\|^2 \, dx 
            = \frac{c}{2} \int_{\R^d} \rho_t \, \Delta  \|\nabla f_t\|^2 \, dx.
        \end{align*}
        \item The term in~\eqref{Eq:Calc1b} is, by distributing the gradient,
        \begin{align*}
            &2 \int_{\R^d} \rho_t \, \langle \nabla f_t, \nabla ( -\langle \nabla f_t, b_t \rangle + \nabla \cdot b_t) \rangle \, dx \\
            &= -2 \int_{\R^d} \rho_t \, \langle \nabla f_t, (\nabla^2 f_t) \, b_t \rangle \, dx 
            - 2 \int_{\R^d} \rho_t \, \langle \nabla f_t, (\nabla b_t) \, \nabla f_t \rangle \, dx 
            +2 \int_{\R^d} \rho_t \langle \nabla f_t, \, \Tr(\nabla^2 b_t) \rangle \, dx.
        \end{align*}
        \item The first term in~\eqref{Eq:Calc1c} is, by Bochner's formula~\eqref{Eq:Bochner}:
        \begin{align*}
            c \int_{\R^d} \rho_t \, \langle \nabla f_t, \nabla (\Delta f_t) \rangle \, dx
            &= \frac{c}{2} \int_{\R^d} \rho_t \, \Delta \|\nabla f_t\|^2 \, dx
            - c \int_{\R^d} \rho_t \, \|\nabla^2 f_t\|^2_{\HS} \, dx.
        \end{align*}
        \item The second term in~\eqref{Eq:Calc1c} is, using relation~\eqref{Eq:Rel2} and integration by parts:
        \begin{align*}
            - c \int_{\R^d} \rho_t \, \langle \nabla f_t, \nabla \|\nabla f_t\|^2 \rangle \, dx
            &= c \int_{\R^d} \langle \nabla \rho_t, \nabla \|\nabla f_t\|^2 \rangle \, dx 
            = - c \int_{\R^d} \rho_t \, \Delta \|\nabla f_t\|^2 dx.
        \end{align*}
    \end{enumerate}
    Combining the above, we see that the terms involving $\langle b_t, \, (\nabla^2 f_t) \, \nabla f_t \rangle$ and $\Delta \|\nabla f_t\|^2$ vanish, and we are left with:
    \begin{align}
        \blue{\frac{d}{dt} \E_{\rho_t}\left[\|\nabla f_t\|^2 \right] 
        = -c \, \E_{\rho_t}\left[\left\|\nabla^2 f_t \right\|^2_{\HS} \right] 
        -2 \E_{\rho_t}\left[\langle \nabla f_t, (\nabla b_t) \, \nabla f_t \rangle  \right] 
        + 2 \E_{\rho_t}\left[\langle \nabla f_t, \, \Tr(\nabla^2 b_t) \rangle  \right].} \label{Eq:FirstTerm-1}
    \end{align}

    %%%%%%%%%%%
    \paragraph{(II) Second term:} We compute the time derivative on the second term in~\eqref{Eq:FIt}, using the Fokker-Planck equation~\eqref{Eq:FPrhot} and the formula~\eqref{Eq:gt}:
    \begin{subequations}
    \begin{align}
        \frac{d}{dt} \E_{\rho_t}\left[\|\nabla g_t\|^2 \right]
        &= \int_{\R^d} (\partial_t \rho_t) \|\nabla g_t\|^2 \, dx \, + \,  2 \int_{\R^d} \rho_t \, \langle \nabla g_t, \nabla (\partial_t g_t) \rangle \, dx \notag \\
        &= -\int_{\R^d} \nabla \cdot (\rho_t b_t) \, \|\nabla g_t\|^2 \, dx  + \frac{c}{2} \int_{\R^d} \Delta \rho_t \, \|\nabla g_t\|^2 \, dx 
        \label{Eq:Calc2a}  \\
        &\qquad
        + 2 \int_{\R^d} \rho_t \, \langle \nabla g_t, \nabla ( -\langle \nabla g_t, b_t \rangle) \rangle \, dx
        + 2 \int_{\R^d} \rho_t \, \langle \nabla g_t, \nabla (\nabla \cdot b_t) \rangle \, dx
        \label{Eq:Calc2b} \\
        &\qquad + c \int_{\R^d} \rho_t \, \langle \nabla g_t, \nabla \Delta g_t \rangle \, dx 
         - c \int_{\R^d} \rho_t \, \langle \nabla g_t, \nabla \|\nabla g_t\|^2 \rangle \, dx. \label{Eq:Calc2c} 
    \end{align}        
    \end{subequations}
    We calculate the terms above one by one.
    \begin{enumerate}
        \item The first term in~\eqref{Eq:Calc2a} is, by integration by parts:
        \begin{align*}
            -\int_{\R^d} \nabla \cdot (\rho_t b_t) \, \|\nabla g_t\|^2 \, dx
            &= \int_{\R^d} \rho_t \langle b_t, \nabla (\|\nabla g_t\|^2) \rangle \, dx 
            = 2\int_{\R^d} \rho_t \langle b_t, (\nabla^2 g_t) \, \nabla g_t \rangle \, dx.
        \end{align*}
        \item The second term in~\eqref{Eq:Calc2a} is, by integration by parts and relation~\eqref{Eq:Rel1}:
        \begin{align}
            \frac{c}{2} \int_{\R^d} (\Delta \rho_t) \,  \|\nabla g_t\|^2 \, dx 
            &= \frac{c}{2} \int_{\R^d} \rho_t \langle \nabla f_t, \, \nabla \|\nabla g_t\|^2 \rangle \, dx \notag \\
            &=\, c \int_{\R^d} \rho_t \, \langle \nabla f_t, \, (\nabla^2 g_t) \, \nabla g_t \rangle \, dx. \label{Eq:Calc2a-a}
        \end{align}
        \item The first term in~\eqref{Eq:Calc2b} is, by chain rule,
        \begin{align*}
            -2 \int_{\R^d} \rho_t \, \langle \nabla g_t, \nabla ( \langle \nabla g_t, b_t \rangle) \rangle \, dx 
            &= -2 \int_{\R^d} \rho_t \, \langle \nabla g_t, (\nabla^2 g_t) \, b_t \rangle \, dx
            \,-\, 2 \int_{\R^d} \rho_t \, \langle \nabla g_t, (\nabla b_t) \, \nabla g_t \rangle \, dx.
        \end{align*}
        \item The second term in~\eqref{Eq:Calc2b} is:
        \begin{align*}
            2 \int_{\R^d} \rho_t \, \langle \nabla g_t, \nabla (\nabla \cdot b_t) \rangle \, dx 
            &= 2 \int_{\R^d} \rho_t \langle \nabla g_t, \, \Tr(\nabla^2 b_t) \rangle \, dx \\
            &= 2 \int_{\R^d} \rho_t \langle \nabla f_t, \, (\nabla b_t) \nabla g_t \rangle \, dx
            \,-\, 2 \int_{\R^d} \rho_t \langle \nabla b_t, \, \nabla^2 g_t \rangle_{\HS} \, dx
        \end{align*}
        where the last equality above follows from the relation~\eqref{Eq:Rel2} and integration by parts:
        \begin{align*}
            \int_{\R^d} \rho_t \langle \nabla f_t, \, (\nabla b_t) \nabla g_t \rangle \, dx
            &= -\int_{\R^d} \langle \nabla \rho_t, \, (\nabla b_t) \nabla g_t \rangle \, dx \\
            &= \int_{\R^d} \rho_t \, \nabla \cdot ((\nabla b_t) \nabla g_t) \, dx \\
            &= \int_{\R^d} \rho_t \, \left( \langle \Tr(\nabla^2 b_t), \, \nabla g_t \rangle + \langle \nabla b_t, \, \nabla^2 g_t \rangle_{\HS} \right) dx.
        \end{align*}
        \item The first term in~\eqref{Eq:Calc2c} is, by Bochner's formula~\eqref{Eq:Bochner} and the same calculation as in term~\eqref{Eq:Calc2a-a}:
        \begin{align*}
            c \int_{\R^d} \rho_t \, \langle \nabla g_t, \nabla \Delta g_t \rangle \, dx 
            &= \frac{c}{2} \int_{\R^d} \rho_t \, \Delta \|\nabla g_t\|^2 \, dx
            - c\int_{\R^d} \rho_t \, \|\nabla^2 g_t\|^2_{\HS} \, dx \\
            &= c \int_{\R^d} \rho_t \left(\left\langle \nabla f_t, \, (\nabla^2 g_t) \, \nabla g_t \right\rangle - \,\|\nabla^2 g_t\|^2_{\HS} \right) \, dx.
        \end{align*}
        \item The second term in~\eqref{Eq:Calc2c} is:
        \begin{align*}
            - c\int_{\R^d} \rho_t \, \langle \nabla g_t, \nabla \|\nabla g_t\|^2 \rangle \, dx
            &= -2c \int_{\R^d} \rho_t \, \left\langle \nabla g_t, \, (\nabla^2 g_t) \, \nabla g_t \right\rangle \, dx.
        \end{align*}
    \end{enumerate}
    Combining, we see the terms involving $\langle b_t, (\nabla^2 g_t) \, \nabla g_t \rangle$ vanish, and we obtain:
    \begin{subequations}\label{Eq:SecondTerm-1}
        \begin{align}
            &\red{\frac{d}{dt} \E_{\rho_t}\left[\|\nabla g_t\|^2 \right]} \notag \\
            &\red{= 2c \, \E_{\rho_t} \left[ \langle \nabla f_t, \, (\nabla^2 g_t) \, \nabla g_t \rangle \right] 
            - 2c \, \E_{\rho_t} \left[ \left\langle \nabla g_t, \, (\nabla^2 g_t) \, \nabla g_t \right\rangle \right] 
            - c \, \E_{\rho_t} \left[\|\nabla^2 g_t\|^2_{\HS} \right]} \\
            &\quad \red{- 2 \E_{\rho_t} \left[ \langle \nabla g_t, (\nabla b_t) \, \nabla g_t \rangle \right] + 2 \E_{\rho_t} \left[ \langle \nabla f_t, \, (\nabla b_t) \nabla g_t \rangle \right] - 2 \E_{\rho_t} \left[\langle \nabla b_t, \, \nabla^2 g_t \rangle_{\HS} \right].}
        \end{align}        
    \end{subequations}

    %%%%%%%%%%%
    \paragraph{(III) Third term:} We compute the time derivative on the third term in~\eqref{Eq:FIt}, using the Fokker-Planck equation~\eqref{Eq:FPrhot} and the formula~\eqref{Eq:gt}, and integration by parts, to get:
    \begin{subequations}
    \begin{align}
        \frac{d}{dt} \left(-2 \E_{\rho_t}\left[\Delta g_t \right] \right)
        &= -2 \int_{\R^d} (\partial_t \rho_t) \, \Delta g_t \, dx \, - 2 \int_{\R^d} \rho_t \, \Delta (\partial_t g_t)  \, dx \notag \\
        &=  2 \int_{\R^d} \nabla \cdot (\rho_t b_t) \, (\Delta g_t) \, dx  
        -c \int_{\R^d} (\Delta \rho_t) \, (\Delta g_t) \, dx  \label{Eq:Calc3a} \\
        &\qquad +2 \int_{\R^d} \rho_t \, \Delta (\langle \nabla g_t, b_t \rangle - \nabla \cdot b_t)  \, dx   \label{Eq:Calc3b} \\
        &\qquad - c \int_{\R^d} \rho_t \, \Delta \Delta g_t \, dx 
        + c \int_{\R^d} \rho_t \, \Delta \|\nabla g_t\|^2 \, dx.   \label{Eq:Calc3c}
    \end{align}
    \end{subequations}
    We calculate the terms above one by one.
    \begin{enumerate}
        \item The first term in~\eqref{Eq:Calc3a} is, by integration by parts:
        \begin{align*}
            2 \int_{\R^d} \nabla \cdot (\rho_t b_t) \, (\Delta g_t) \, dx  
            &= -2 \int_{\R^d} \rho_t \, \langle b_t, \, \nabla \Delta g_t \rangle \, dx.  
        \end{align*}
        \item For the second term in~\eqref{Eq:Calc3a}, we can write it using integration by parts as follows:
        \allowdisplaybreaks
        \begin{align}
            -c \int_{\R^d} (\Delta \rho_t) \, (\Delta g_t) \, dx
            &= -c \int_{\R^d} \left(\sum_{i=1}^d \frac{\partial^2}{\partial x_i^2} \rho_t(x) \right) \, \left(\sum_{j=1}^d \frac{\partial^2}{\partial x_j^2} g_t(x) \right) \, dx \notag \\
            &= -c \int_{\R^d} \sum_{i,j=1}^d  \left(\frac{\partial^2}{\partial x_i^2} \rho_t(x) \right) \, \left(\frac{\partial^2}{\partial x_j^2} g_t(x) \right) \, dx \notag \\
            &= -c \int_{\R^d} \sum_{i,j=1}^d  \rho_t(x) \, \left(\frac{\partial^4}{\partial x_i^2 \, \partial x_j^2} g_t(x) \right) \, dx \notag \\
            &= -c \int_{\R^d} \sum_{i,j=1}^d \left(\frac{\partial^2}{\partial x_i \, \partial x_j} \rho_t(x) \right) \left(\frac{\partial^2}{\partial x_i \, \partial x_j} g_t(x) \right) \, dx \notag \\
            &= -c \int_{\R^d} \langle \nabla^2 \rho_t, \nabla^2 g_t \rangle_{\HS} \, dx \notag \\
            &= c \, \int_{\R^d} \rho_t \, \left( \langle \nabla^2 f_t, \nabla^2 g_t \rangle_{\HS} - \left\langle \nabla f_t, \, (\nabla^2 g_t) \, \nabla f_t \right\rangle \right) \, dx \label{Eq:Calc3a-2}
        \end{align}
        where in the last step we have used relation~\eqref{Eq:Rel2}.
        \item The term in~\eqref{Eq:Calc3b} is, by the relation~\eqref{Eq:Rel1}, integration by parts, and distributing the gradient:
        \begin{align*}
            &2 \int_{\R^d} \rho_t \, \Delta ( \langle \nabla g_t, b_t \rangle - \nabla \cdot b_t)  \, dx \\
            &= 2 \int_{\R^d} \rho_t \, \langle \nabla f_t, \, \nabla ( \langle \nabla g_t, b_t \rangle - \nabla \cdot b_t) \rangle  \, dx \\
            &= 2 \int_{\R^d} \rho_t \left( \langle \nabla f_t, \, (\nabla^2 g_t) \, b_t \rangle + \langle \nabla f_t, \, (\nabla b_t) \, \nabla g_t \rangle - \langle \nabla f_t, \, \Tr(\nabla^2 b_t) \rangle \right) dx \\
            &= \int_{\R^d} \rho_t \, \left( 2\langle \nabla \Delta g_t, b_t \rangle +2 \langle \nabla^2 g_t, \nabla b_t \rangle_{\HS} 
            +2 \langle \nabla f_t, \, (\nabla b_t) \, \nabla g_t \rangle
            - 2 \langle \nabla f_t, \, \Tr(\nabla^2 b_t) \rangle\right) \, dx
        \end{align*}
        where in the last equality we have used the integration by parts to write:
        \begin{align*}
            \int_{\R^d} \rho_t \, \langle \nabla f_t, \, (\nabla^2 g_t) \, b_t \rangle \, dx 
            &= -\int_{\R^d} \langle \nabla \rho_t, \, (\nabla^2 g_t) \, b_t \rangle \, dx  \\
            &= \int_{\R^d} \rho_t \, \nabla \cdot \left((\nabla^2 g_t) \, b_t \right) \, dx \\
            &= \int_{\R^d} \rho_t \, \left(\langle \nabla \Delta g_t, b_t \rangle + \langle \nabla^2 g_t, \nabla b_t \rangle_{\HS} \right) dx            
        \end{align*}
        and we have also used the identity $\nabla \cdot (\nabla^2 g_t) = \nabla \Delta g_t$, which can be verified from the definition.
        \item The first term in~\eqref{Eq:Calc3c} is, by integration by parts, equal to the second term in~\eqref{Eq:Calc3a}, which is also equal to~\eqref{Eq:Calc3a-2}:
        \begin{align*}
             -c \int_{\R^d} \rho_t \, \Delta \Delta g_t \, dx
             \,&=\, - c \int_{\R^d} (\Delta \rho_t) \, (\Delta g_t) \, dx \\
             \,&=\, c \, \int_{\R^d}\rho_t \, \left( \langle \nabla^2 f_t, \nabla^2 g_t \rangle_{\HS} - \left\langle \nabla f_t, \, (\nabla^2 g_t) \, \nabla f_t \right\rangle \right) dx.
        \end{align*}
        \item The second term in~\eqref{Eq:Calc3c} is, by integration by parts, equal to twice the second term in~\eqref{Eq:Calc2a}, which is equal to twice the term in~\eqref{Eq:Calc2a-a}:
        \begin{align*}
            c \, \int_{\R^d} \rho_t \, \Delta \|\nabla g_t\|^2 \, dx
            = c \, \int_{\R^d} \Delta \rho_t \, \|\nabla g_t\|^2 \, dx
            &= 2c \, \int_{\R^d} \rho_t \, \langle \nabla f_t, \, (\nabla^2 g_t) \, \nabla g_t \rangle \, dx.
        \end{align*}
    \end{enumerate}
    Combining, we see the terms involving $\langle b_t, \, \nabla \Delta g_t \rangle$ vanish, so we get:
    \begin{subequations}\label{Eq:ThirdTerm-1}
    \begin{align}
        \purple{\frac{d}{dt} \left(-2 \E_{\rho_t}\left[\Delta g_t \right] \right)}
        &\purple{= 2c \, \E_{\rho_t}\left[ \langle \nabla^2 f_t, \nabla^2 g_t \rangle_{\HS} \right] 
        - 2c \, \E_{\rho_t}\left[ \left\langle \nabla f_t, \, (\nabla^2 g_t) \, \nabla f_t \right\rangle \right]} \\
        &\quad \purple{\, +2\E_{\rho_t}\left[ \langle \nabla^2 g_t, \nabla b_t \rangle_{\HS} \right] 
        +2\E_{\rho_t}\left[\langle \nabla f_t, \, (\nabla b_t) \, \nabla g_t \rangle \right]} \\
        &\quad \purple{\,- 2\E_{\rho_t}\left[\langle \nabla f_t, \, \Tr(\nabla^2 b_t) \rangle\right]
        + 2c \, \E_{\rho_t} \left[\langle \nabla f_t, \, (\nabla^2 g_t) \, \nabla g_t \rangle \right].}
    \end{align}
    \end{subequations}

    \paragraph{Combining the terms.}
    Combining the calculations in~\eqref{Eq:FirstTerm-1},~\eqref{Eq:SecondTerm-1}, and~\eqref{Eq:ThirdTerm-1}, we see the terms involving $\E_{\rho_t}\left[\langle \nabla f_t, \, \Tr(\nabla^2 b_t) \rangle  \right]$ and  $\E_{\rho_t} \left[\langle \nabla b_t, \, \nabla^2 g_t \rangle_{\HS} \right]$ vanish.
    After rearranging the terms, we find the time derivative of the relative Fisher information~\eqref{Eq:FIt} is equal to:
    \begin{align*}
        &\frac{d}{dt} \FI(\rho_t \,\|\, \nu_t)
        \,=\, \blue{\frac{d}{dt} \E_{\rho_t}\left[\|\nabla f_t\|^2 \right]} 
        + \red{\frac{d}{dt} \E_{\rho_t}\left[\|\nabla g_t\|^2\right]} 
        + \purple{\frac{d}{dt} \left(- 2 \E_{\rho_t}\left[\Delta g_t \right]\right)} \\
        &= \blue{-c \, \E_{\rho_t}\left[\left\|\nabla^2 f_t \right\|^2_{\HS} \right]} 
        \red{\,-\,c \, \E_{\rho_t} \left[\|\nabla^2 g_t\|^2_{\HS} \right]}  
        \purple{\,+\,  2c \, \E_{\rho_t}\left[ \langle \nabla^2 f_t, \nabla^2 g_t \rangle_{\HS} \right]} \\
        &\quad \purple{\,-\, 2c \, \E_{\rho_t}\left[ \left\langle \nabla f_t, \, (\nabla^2 g_t) \, \nabla f_t \right\rangle \right]} 
        \red{\,-\, 2c \, \E_{\rho_t} \left[ \left\langle \nabla g_t, \, (\nabla^2 g_t) \, \nabla g_t \right\rangle \right]} 
        \red{\,+\, 2c \, \E_{\rho_t} \left[ \langle \nabla f_t, \, (\nabla^2 g_t) \, \nabla g_t \rangle \right]} \\
        &\qquad\qquad \purple{\,+\, 2c \, \E_{\rho_t} \left[ \langle \nabla f_t, \, (\nabla^2 g_t) \, \nabla g_t \rangle \right]} \\
        &\quad 
        \blue{\, - \,  2 \E_{\rho_t}\left[\langle \nabla f_t, (\nabla b_t) \, \nabla f_t \rangle  \right]}
        \red{- 2 \E_{\rho_t} \left[ \langle \nabla g_t, (\nabla b_t) \, \nabla g_t \rangle \right]}
        \red{\,+\, 2 \E_{\rho_t} \left[ \langle \nabla f_t, \, (\nabla b_t) \nabla g_t \rangle \right]} \\
        &\qquad\qquad 
         \purple{\,+\, 2 \E_{\rho_t} \left[ \langle \nabla f_t, \, (\nabla b_t) \nabla g_t \rangle \right]} \\
        &= -c \, \E_{\rho_t} \left[\left\| \nabla^2 f_t - \nabla^2 g_t \right\|^2_{\HS} \right] 
        - 2c \, \E_{\rho_t}\left[\left\|\nabla f_t - \nabla g_t \right\|^2_{\nabla^2 g_t} \right]
        - 2\E_{\rho_t}\left[\left\|\nabla f_t - \nabla g_t \right\|^2_{(\nabla b_t)_{\sym}} \right] \\
        &= -c \, \E_{\rho_t} \left[\left\| \nabla^2 f_t - \nabla^2 g_t \right\|^2_{\HS} \right] 
        - 2\E_{\rho_t}\left[\left\|\nabla f_t - \nabla g_t \right\|^2_{(c \, \nabla^2 g_t + (\nabla b_t)_{\sym})} \right]
    \end{align*}
    as desired.
    In the third equality above, we use the property that for all $u,v \in \R^d$ and $A \in \R^{d \times d}$:
    $$\langle u,Av \rangle + \langle v,Au \rangle = \langle u,(A+A^\top)v \rangle = 2\langle u,A_{\sym} v \rangle.$$
\end{proof}

%%%%%%%%%%%%%%%%%%%%%%%
\subsection{Proof of Theorem~\ref{Thm:FisherInfoHeatFlow} (SDPI along Gaussian Channel)}
\label{Sec:FisherInfoHeatFlowProof}

\begin{proof}[Proof of Theorem~\ref{Thm:FisherInfoHeatFlow}]
    For the Gaussian channel~\eqref{Eq:GaussianChannel}, which is the case $b_t = 0$ and $c=1$, the time derivative formula from Lemma~\ref{Lem:TimeDerivativeFI} simplifies to:
    \begin{align}\label{Eq:ddtFIGaussianChannel}
        \frac{d}{dt} \FI(\rho_t \,\|\, \nu_t) = -\E_{\rho_t} \left[\left\| \nabla^2 \log \frac{\rho_t}{\nu_t} \right\|^2_{\HS} \right] - 2\E_{\rho_t}\left[\left\|\nabla \log \frac{\rho_t}{\nu_t} \right\|^2_{(-\nabla^2 \log \nu_t)}\right].
    \end{align}
    The first term in the right-hand side above is the second-order Fisher information, which is non-negative: 
    $\E_{\rho_t} \big[\big\| \nabla^2 \log \frac{\rho_t}{\nu_t} \big\|^2_{\HS} \big] \ge 0$.
    In parts (i), (ii), and (iv) below we drop this first term, and only consider the second term which is a weighted relative Fisher information.
    In part (iii), we use the Poincar\'e inequality to exploit this second-order Fisher information to improve the convergence rate.

    Recall along the Gaussian channel~\eqref{Eq:GaussianChannel}, we have $\rho_t = \rho_0 \ast \N(0, tI)$ and $\nu_t = \nu_0 \ast \N(0, tI)$.

    \paragraph{Part (i):} Note the Gaussian distribution $\N(0, tI)$ is log-concave, and recall convolution with log-concave distributions preserves log-concavity~\cite[Proposition~3.5]{SW14}.
    Since $\nu_0$ is log-concave, $\nu_t = \nu_0 \ast \N(0, tI)$ is also log-concave, so $-\nabla^2 \log \nu_t \succeq 0$.
    Then from the formula~\eqref{Eq:ddtFIGaussianChannel} by dropping the first term:
    $$\frac{d}{dt} \FI(\rho_t \,\|\, \nu_t) \le -2\E_{\rho_t}\left[\left\|\nabla \log \frac{\rho_t}{\nu_t} \right\|^2_{(-\nabla^2 \log \nu_t)}\right] \le 0.$$
    This shows that relative Fisher information is non-increasing over time:
    $\FI(\rho_t \,\|\, \nu_t) \le \FI(\rho_0 \,\|\, \nu_0).$

    \paragraph{Part (ii):} Now assume $\nu_0$ is $\alpha$-SLC. 
    Since $\N(0, tI)$ is $(1/t)$-SLC, we recall $\nu_t = \nu_0 \ast \N(0, tI)$ is $\alpha_t = \frac{\alpha}{1+\alpha t}$-SLC~\cite[Theorem~3.7]{SW14}.
    Then from the formula~\eqref{Eq:ddtFIGaussianChannel} by dropping the first term:
    \begin{align}\label{Eq:ddtFIGaussianChannel2}
        \frac{d}{dt} \FI(\rho_t \,\|\, \nu_t) \le -2\alpha_t \, \E_{\rho_t}\left[\left\|\nabla \log \frac{\rho_t}{\nu_t} \right\|^2\right] = -2\alpha_t \, \FI(\rho_t \,\|\, \nu_t).
    \end{align}
    Integrating and plugging in the definition of $\alpha_t = \frac{\alpha}{1 + \alpha t} = \frac{d}{dt} \log(1 + \alpha t)$ give the desired result:
    \begin{align*}
    \FI(\rho_t \,\|\, \nu_t) 
        \,\le\, \exp\left(-2\int_0^t \alpha_s \, ds\right) \FI(\rho_0 \,\|\, \nu_0) \,=\, \frac{\FI(\rho_0 \,\|\, \nu_0)}{(1+\alpha t)^2}.
    \end{align*}

    \paragraph{Part (iii):}
    As in the proof of part (ii) above, $\nu_t = \nu_0 \ast \N(0, tI)$ is $\alpha_t = \frac{\alpha}{1+t\alpha}$-SLC.
    Furthermore, since $\N(0,tI)$ is $(1/t)$-SLC and thus satisfies $(1/t)$-Poincar\'e inequality, and $\rho_0$ satisfies $\beta$-Poincar\'e inequality by assumption, we know $\rho_t = \rho_0 \ast \N(0, tI)$ satisfies $\beta_t = \frac{\beta}{1+\beta t}$-Poincar\'e inequality~\cite[Corollary~3.1]{chafai2004entropies}.
    
    Since $\rho_0$ and $\nu_0$ are symmetric, $\rho_t = \rho_0 \ast \N(0, tI)$ and $\nu_t = \nu_0 \ast \N(0, tI)$ are also symmetric.
    In particular, $\log \nu_t(-x) = \log \nu_t(x)$, so $\nabla \log \nu_t$ is an odd function: $\nabla \log \nu_t(-x) = -\nabla \log \nu_t(x)$.
    Furthermore, since $\rho_t$ is symmetric ($\rho_t(-x) = \rho_t(x)$), this implies $\E_{\rho_t}[\nabla \log \nu_t] = 0$, since indeed:
    \begin{align*}
        \E_{\rho_t}[\nabla \log \nu_t]
        &= \int_{\R^d} \rho_t(x) \nabla \log \nu_t(x) \, dx \\
        &= \frac{1}{2} \int_{\R^d} \rho_t(x) \nabla \log \nu_t(x) \, dx + \frac{1}{2} \int_{\R^d} \rho_t(-x) \nabla \log \nu_t(-x) \, dx \\
        &= \frac{1}{2} \int_{\R^d} \rho_t(x) \left(\nabla \log \nu_t(x) + \nabla \log \nu_t(-x)\right) \, dx \\
        &= 0.
    \end{align*}
    Since $\rho_t = \rho_0 \ast \N(0,tI)$ has full support on $\R^d$,
    % and the density has fast decay at infinity; then 
    by integration by parts, $\E_{\rho_t}[\nabla \log \rho_t] = 0$.
    Combining, this shows that $\E_{\rho_t}[\nabla \log \frac{\rho_t}{\nu_t}] = 0$.
    Therefore, by the Poincar\'e inequality for $\rho_t$ applied to the function $\phi = \nabla \log \frac{\rho_t}{\nu_t}$,
    \begin{align*}
        \E_{\rho_t} \left[\left\| \nabla^2 \log \frac{\rho_t}{\nu_t} \right\|^2_{\HS} \right]
        &\ge \beta_t \, \Var_{\rho_t}\left(\nabla \log \frac{\rho_t}{\nu_t}\right) 
        = \beta_t \, \E_{\rho_t}\left[\left\|\nabla \log \frac{\rho_t}{\nu_t}\right\|^2\right] 
        = \beta_t \, \FI(\rho_t \,\|\, \nu_t).
    \end{align*}
    Therefore, by the formula~\eqref{Eq:ddtFIGaussianChannel}:
    \begin{align*}
        \frac{d}{dt} \FI(\rho_t \,\|\, \nu_t) 
        &= -\E_{\rho_t} \left[\left\| \nabla^2 \log \frac{\rho_t}{\nu_t} \right\|^2_{\HS} \right] - 2\E_{\rho_t}\left[\left\|\nabla \log \frac{\rho_t}{\nu_t} \right\|^2_{(-\nabla^2 \log \nu_t)}\right] \\
        &\le -\left(\beta_t + 2\alpha_t \right) \, \FI(\rho_t \,\|\, \nu_t).
    \end{align*}
    Integrating and plugging in the definitions of $\alpha_t = \frac{\alpha}{1+\alpha t}$ and $\beta_t = \frac{\beta}{1+\beta t}$ give the desired result:
    \begin{align*}
    \FI(\rho_t \,\|\, \nu_t) 
        \,\le\, \exp\left(-\int_0^t (\beta_s + 2\alpha_s) \, ds\right) \FI(\rho_0 \,\|\, \nu_0) \,=\, \frac{\FI(\rho_0 \,\|\, \nu_0)}{(1+\beta t)(1+\alpha t)^2}.
    \end{align*}

    \paragraph{Part (iv):}
    By the result of~\cite[Theorem~1.3]{brigati2024heat}, under the assumption on $\nu_0$, we have that $\nu_t = \nu_0 \ast \N(0,tI)$ satisfies
    $-\nabla^2 \log \nu_t(x) \succeq \alpha_t I$ for all $x \in \R^d$, where
    $$\alpha_t := \frac{1}{t}\left(1-\frac{1}{t} \left(\frac{L}{\alpha+1/t} + \frac{1}{\sqrt{\alpha + 1/t}}\right)^2\right).$$
    Note $\alpha_t$ may be negative for small $t > 0$, but we have $\alpha_t > 0$ for instance for all $t \ge \left(\frac{L}{\alpha} + \frac{1}{\sqrt{\alpha}}\right)^2$.
    By direct computation, we may observe that:
    $$\alpha_t = \frac{d}{dt} \left(\log (\alpha t+1) - \frac{L^2}{\alpha +1/t} - \frac{4L}{\sqrt{\alpha+1/t}}\right).$$
    By the same computation~\eqref{Eq:ddtFIGaussianChannel2} as in part (ii) above, which holds for $\alpha_t \in \R$, we have for all $t \ge 0$:
    \begin{align*}
    \FI(\rho_t \,\|\, \nu_t) 
        \,\le\, \exp\left(-2\int_0^t \alpha_s \, ds\right) \FI(\rho_0 \,\|\, \nu_0) \,=\, \frac{\FI(\rho_0 \,\|\, \nu_0)}{(1+\alpha t)^2} \, \exp\left(\frac{2tL^2}{\alpha t+1} + \frac{8L \sqrt{t}}{\sqrt{\alpha t+1}}\right).
    \end{align*}
    
\end{proof}

%%%%%%%%
\subsection{Detail for Counterexample of DPI along the Gaussian Channel}
\label{Sec:CounterexampleDecayGaussianChannelProof}

The following is a detailed computation for Example~\ref{Ex:CounterDecrHeat}.
See Figure~\ref{Fig:CounterexampleDecreaseHeat} for an illustration.

\begin{proposition}
\label{Prop:CounterexampleDecayGaussianChannel}
    Let $\rho_0 = \N(0,1)$ on $\R$.
    Let $M \ge 2$ and $L \ge 2$ be arbitrary, and define $\nu_0 \propto e^{-g}$ where $g \colon \R \to \R$ is the function:
    \begin{align}\label{Eq:BadInit}
        g(x) = 
        \begin{cases}
            -\frac{M}{2} x^2 ~~~ & \text{ if } |x| \le L, \\
           \frac{1}{2} (x-L)^2 - M L(x-L) - \frac{M L^2}{2}  & \text{ if } x > L, \\
            \frac{1}{2} (x+L)^2 + M L(x+L) - \frac{M L^2}{2}  & \text{ if } x < -L.
        \end{cases}
    \end{align}
    For $t \ge 0$, define $\rho_t = \rho_0 \ast \N(0,t)$ and $\nu_t = \nu_0 \ast \N(0,t)$, which are the output of the Gaussian channel~\eqref{Eq:GaussianChannel} at time $t$.
    Then for small $t > 0$, the relative Fisher information $t \mapsto \FI(\rho_t\,\|\,\nu_t)$ is increasing, and in particular,
    $$\FI(\rho_t\,\|\,\nu_t) > \FI(\rho_0\,\|\,\nu_0).$$
\end{proposition}
\begin{proof}
    We show that the time derivative of $t \mapsto \FI(\rho_t\,\|\,\nu_t)$ is positive at $t = 0$; then by continuity, the time derivative will still be positive for sufficiently small $t > 0$, which means $t \mapsto \FI(\rho_t\,\|\,\nu_t)$ is increasing.
    Since we are in $1$-dimension, we denote derivative by $'$ and second derivative by $''$.

    Recall by the time derivative formula~\eqref{Eq:ddtFIGaussianChannel} along the Gaussian channel, we have:
    \begin{align}
        \frac{d}{dt} \FI(\rho_t\,\|\,\nu_t) \Big|_{t = 0} 
        &= -\E_{\rho_0}\left[\left(\left(\log \frac{\rho_0}{\nu_0}\right)^{''}\right)^2\right] - 2\E_{\rho_0}\left[(-\log \nu_0)'' \, \left(\left(\log \frac{\rho_0}{\nu_0}\right)^{'}\right)^2 \right] \notag \\
        &= -\E_{\rho_0}\left[\left(-1 + g''(X) \right)^2 \right] - 2\E_{\rho_0}\left[g''(X) \, (-X+g'(X))^2 \right]. \label{Eq:ExCounterCalc1}
    \end{align}
    By construction, we have:
    \begin{align}\label{Eq:Formulag1}
        g'(x) &= \begin{cases}
            -M x ~~~ & \text{ if } |x| \le L, \\
            x-(M+1) L & \text{ if } x > L, \\
            x+(M+1) L & \text{ if } x < -L
        \end{cases}
    \end{align}
    and
    \begin{align}\label{Eq:Formulag2}
        g''(x) &= \begin{cases}
            -M ~~~ & \text{ if } |x| \le L, \\
            1 & \text{ if } |x| > L.
        \end{cases}
    \end{align}
    Then we can bound the first term in~\eqref{Eq:ExCounterCalc1} by:
    \begin{align*}
        -\E_{\rho_0}\left[\left(-1 + g''(X) \right)^2 \right]
        &= -(M+1)^2 \, \rho_0([-L,L]) 
        \,\ge\, -(M+1)^2.
    \end{align*}
    Furthermore, with $\rho_0 = \N(0,1)$, we can compute the second term in~\eqref{Eq:ExCounterCalc1}:
    \begin{align*}
        - 2\E_{\rho_0}\left[g''(X) \, (-X+g'(X))^2 \right]
        = 2 M (M+1)^2 \E_{\rho_0}[X^2 \, \mathbf{1}\{|X| \le L\}] - 4(M+1)^2 L^2 \rho_0((L,\infty)).
    \end{align*}
    We choose $L > 0$ such that $\E_{\rho_0}[X^2 \, \mathbf{1}\{|X| \le L\}] > \frac{1}{2} \E_{\rho_0}[X^2] = \frac{1}{2}$ and $L^2 \rho_0((L,\infty)) < \frac{1}{4}$; note that any $L \ge 2$ works.
    Then combining the bounds above in~\eqref{Eq:ExCounterCalc1}, we get:
    \begin{align*}
        \frac{d}{dt} \FI(\rho_t\,\|\,\nu_t) \Big|_{t = 0} 
        &= -\E_{\rho_0}\left[\left(-1 + g''(X) \right)^2 \right] - 2\E_{\rho_0}\left[g''(X) \, (-X+g'(X))^2 \right] \\
        &> -(M+1)^2 + M(M+1)^2 - (M+1)^2 \\
        &= (M-2)(M+1)^2 \\
        &\ge 0
    \end{align*}
    where the last inequality holds since $M \ge 2$.  
\end{proof}

%%%%%%%%%%%%%%%%%%%%%%%
\subsection{Proof of Theorem~\ref{Thm:FisherInfoOUFlow} (Eventual SDPI along the OU Channel)}
\label{Sec:SDPI-OU-Proof}

\begin{proof}[Proof of Theorem~\ref{Thm:FisherInfoOUFlow}]
    For the OU channel~\eqref{Eq:OU}, we have $b_t(x) = \gamma x$ and $c=2$, so the time derivative formula from Lemma~\ref{Lem:TimeDerivativeFI} becomes:
    \begin{align}\label{Eq:ddtFIOUChannel}
        \frac{d}{dt} \FI(\rho_t \,\|\, \nu_t) = -2\E_{\rho_t} \left[\left\| \nabla^2 \log \frac{\rho_t}{\nu_t} \right\|^2_{\HS} \right] - 2\E_{\rho_t}\left[\left\|\nabla \log \frac{\rho_t}{\nu_t} \right\|^2_{(-2\nabla^2 \log \nu_t - \gamma I)}\right].
    \end{align}
    The first term in the right-hand side above is the second-order Fisher information, which is non-negative: 
    $\E_{\rho_t} \big[\big\| \nabla^2 \log \frac{\rho_t}{\nu_t} \big\|^2_{\HS} \big] \ge 0$.
    In parts (i) we drop this first term, and only consider the second term which is a weighted relative Fisher information.
    In part (ii), we apply the Poincar\'e inequality to exploit this second-order Fisher information to improve the convergence rate.

    \medskip
    \noindent
    \textbf{Part (i):} 
    We recall the explicit solution~\eqref{Eq:OU} of OU at time $t \ge 0$ is $X_t \stackrel{d}{=} e^{-\gamma t} X_0 + \sqrt{\frac{1-e^{-2\gamma t}}{\gamma}} \, Z$ where $X_0 \sim \nu_0$, $Z \sim \N(0,I)$ is independent, and $X_t \sim \nu_t$.
    Since $\nu_0$ is $\alpha$-SLC by assumption, the law of $e^{-\gamma t} X_0$ is $(e^{2\gamma t} \, \alpha)$-SLC.
    Note the law of $\sqrt{\frac{1-e^{-2\gamma t}}{\gamma}} \, Z$ is $(\frac{\gamma}{1-e^{-2\gamma t}})$-SLC.
    Then from the preservation of strong log-concavity under convolution~\cite[Theorem~3.7]{SW14}, we know $\nu_t$ is $\alpha_t$-SLC, where
    $$\alpha_t = \left(\frac{e^{-2\gamma t}}{\alpha} + \frac{1-e^{-2\gamma t}}{\gamma}\right)^{-1} = \frac{\gamma \alpha}{\alpha + e^{-2\gamma t}(\gamma-\alpha)}.$$
    Note that
    \begin{align}\label{Eq:DefLambdat}
        \gamma_t := 2\alpha_t - \gamma 
        \,=\, \gamma - \frac{2\gamma (\gamma-\alpha) e^{-2\gamma t}}{\alpha + e^{-2\gamma t}(\gamma-\alpha)}
        = \frac{d}{dt} \left(\gamma t + \log\left(\alpha + e^{-2\gamma t}(\gamma-\alpha)\right)\right).
    \end{align}
    Then in this case, 
    $-2\nabla^2 \nu_t - \gamma I \succeq (2\alpha_t-\gamma) I = \lambda_t I$, and thus from the formula~\eqref{Eq:ddtFIOUChannel}, we obtain
    \begin{align*}
        \frac{d}{dt} \FI(\rho_t \,\|\, \nu_t) 
        &\le - 2\E_{\rho_t}\left[\left\|\nabla \log \frac{\rho_t}{\nu_t} \right\|^2_{(-2\nabla^2 \log \nu_t - \gamma I)}\right] 
        \le -2\lambda_t \, \FI(\rho_t \,\|\, \nu_t).
    \end{align*}
    Integrating the differential inequality above 
    yields the desired bound:
    \begin{align*}
        \FI(\rho_t \,\|\, \nu_t)
        \le \exp\left(-2\int_0^t \lambda_s \, ds\right) \FI(\rho_0 \,\|\, \nu_0)
        = \frac{\gamma^2 \, e^{-2\gamma t}}{(\alpha + e^{-2\gamma t}(\gamma-\alpha))^2} \, \FI(\rho_0 \,\|\, \nu_0).
    \end{align*}

    \medskip
    \noindent
    \textbf{Part (ii):} Since $\rho_0$ and $\nu_0$ are symmetric, along the OU channel~\eqref{Eq:OU} to $\N(0,\gamma^{-1}I)$, the solutions $\rho_t$ and $\nu_t$ are also symmetric, as can be seen from the explicit solution~\eqref{Eq:OUSol}.    
    Then as in the proof of Theorem~\ref{Thm:FisherInfoHeatFlow}(iii), we have $\E_{\rho_t}[\nabla \log \frac{\rho_t}{\nu_t}] = 0$.
    Furthermore, similar to the calculation in part (i) and using the composition rule for Poincar\'e inequality under convolution~\citep{chafai2004entropies}, since $\rho_0$ satisfies $\beta$-Poincar\'e inequality, along the OU channel~\eqref{Eq:OU}, $\rho_t$ satisfies $\beta_t$-Poincar\'e inequality where:
    $$\beta_t = \frac{\gamma \beta}{\beta + e^{-2\gamma t}(\gamma-\beta)} = \frac{d}{dt} \left( \gamma t + \frac{1}{2} \log\left( \beta + e^{-2\gamma t}(\gamma-\beta)\right)\right).$$
    Then by the Poincar\'e inequality for $\rho_t$, for the test function $\nabla \log \frac{\rho_t}{\nu_t}$, we have:
    \begin{align*}
        \E_{\rho_t} \left[\left\| \nabla^2 \log \frac{\rho_t}{\nu_t} \right\|^2_{\HS} \right]
        &\ge \beta_t \, \Var_{\rho_t}\left(\nabla \log \frac{\rho_t}{\nu_t}\right) 
        = \beta_t \, \E_{\rho_t}\left[\left\|\nabla \log \frac{\rho_t}{\nu_t}\right\|^2\right] 
        = \beta_t \, \FI(\rho_t \,\|\, \nu_t).
    \end{align*}
    We also recall our estimate from part (i) that $-2\nabla^2 \nu_t - \gamma I \succeq \lambda_t I$ where $\lambda_t$ is given in~\eqref{Eq:DefLambdat}.
    Then by the identity~\eqref{Eq:ddtFIOUChannel} along the OU channel~\eqref{Eq:OU}, we obtain:
    \begin{align*}
        \frac{d}{dt} \FI(\rho_t \,\|\, \nu_t) 
        &\le -2\left(\beta_t + \lambda_t \right) \, \FI(\rho_t \,\|\, \nu_t).
    \end{align*}
    Integrating the differential inequality above gives the desired result:
    \begin{align*}
    \FI(\rho_t \,\|\, \nu_t) 
        &\le\, \exp\left(-2\int_0^t (\beta_s + \lambda_s) \, ds\right) \FI(\rho_0 \,\|\, \nu_0) \\
        &=\, \frac{\gamma^3 \, e^{-4\gamma t}}{(\beta + e^{-2\gamma t} (\gamma-\beta))(\alpha+e^{-2 \gamma t}(\gamma-\alpha))^2} \, \FI(\rho_0 \,\|\, \nu_0).
    \end{align*}
\end{proof}

%%%%%%%%%%%%%%%%%%%%%
\section{Proofs for Section~\ref{Sec:ProximalSamplerAnalysis} (Analysis of the Proximal Sampler)}

%%%%%%%%%%%%%%%%%%%%
\subsection{Review of the Rejection Sampling Implementation of the RGO}
\label{Sec:RejectionReview}

We review the algorithmic implementation of the restricted Gaussian oracle (RGO), which is required in each step of the Proximal Sampler (see Section~\ref{Sec:ProximalSamplerReview}).
Here we follow the simple rejection sampling implementation of the RGO from~\cite[Section~4.2]{CCSW22} for the setting when the original target distribution $\nu^X \propto e^{-g}$ is log-smooth.
We note that one can also implement the RGO using inexact rejection sampling to obtain an improved dimension dependence~\citep{fan2023improved}, and other works have shown how to implement the RGO for broader classes of distributions, see e.g.~\citep{gopi2022private,LC22,liang2023a,liang2024proximal}

Recall that the RGO~\eqref{Eq:RGO} needs to sample from the conditional distribution $\nu^{X \mid Y}$ with density, for each fixed $y \in \R^d$:
$$\nu^{X \mid Y}(x \mid y) \propto_x \exp \left( -g(x) - \frac{1}{2\eta} \|x-y\|^2 \right).$$
Define $f_y \colon \R^d \to \R$ by $f_y(x) :=  g(x) + \frac{1}{2\eta} \|x-y\|^2$, so that $\nu^{X \mid Y}(x \mid y) \propto \exp(-f_y(x))$.
Suppose that the potential function $g$ for the original target distribution $\nu^X \propto e^{-g}$ is $L$-smooth, which means $-LI \preceq \nabla^2 g(x) \preceq LI$ for all $x \in \R^d$.
If $0 < \eta < \frac{1}{L}$, then we observe that $f_y$ is an $(\frac{1}{\eta}-L)$-strongly convex and $(\frac{1}{\eta}+L)$-smooth function, with condition number $\kappa := \frac{1+\eta L}{1-\eta L}$.

We can perform rejection sampling on $\nu^{X \mid Y}(\cdot \mid y) \propto \exp(-f_y)$ via the following procedure:
\begin{enumerate}
    \item Compute the minimizer $x^*_y = \arg\min_{x \in \R^d} f_y(x)$; we can do this, e.g., via gradient descent.
    \item Repeat until acceptance: draw a sample $Z \sim \N\left(x^*_y, \frac{\eta}{1-\eta L} I\right)$ and accept it with probability $\exp\left(-f_y(Z) + f_y(x^*_y) + \frac{\eta}{2(1-\eta L)} \|Z-x^*_y\|^2\right)$.
\end{enumerate}
By standard properties of rejection sampling, the output of the procedure above is an exact sample $Z \sim \nu^{X \mid Y}(\cdot \mid y)$, and the expected number of iterations (draws of $Z$) until acceptance is at most $\kappa^{d/2} = \left(\frac{1+\eta L}{1-\eta L}\right)^{d/2}$; see~\cite[Theorem~7]{chewi2022query}.
Note that if we choose $\eta = \frac{1}{dL}$, then the expected number of iterations until acceptance is at most $\kappa^{d/2} = \left(\frac{1+ 1/d}{1- 1/d}\right)^{d/2} = O(1)$.
This shows that we can implement the RGO via rejection sampling with $O(1)$ expected iteration complexity.

%%%%%%%%%%%%%%%%%%%%
\subsection{Proof of Theorem~\ref{Thm:FIProximalRate} (Convergence of Proximal Sampler in Relative Fisher Information)}
\label{Sec:FIProximalRateProof}

\begin{proof}[Proof of Theorem~\ref{Thm:FIProximalRate}]
    As we reviewed in Section~\ref{Sec:ProximalSamplerReviewKL}, we use stochastic process interpretations of the forward and backward steps of the Proximal Sampler.

    \medskip
    \noindent
    \textbf{(1) Forward step:}
    The first step of the Proximal sampling algorithm is a convolution with the Gaussian distribution: $\rho_k^Y = \rho_k^X \ast \N(0, \eta I)$, 
    and $\nu^Y = \nu^X \ast \N(0, \eta I)$.
    These are the outputs of the Gaussian channel~\eqref{Eq:GaussianChannel} for time $\eta$.
    Since $\nu^X$ is $\alpha$-SLC by assumption, by the SDPI along the Gaussian channel from Theorem~\ref{Thm:FisherInfoHeatFlow}(ii), we have:
    $$\FI(\rho_k^Y \,\|\, \nu^Y) \,\le\, \frac{\FI(\rho_k^X \,\|\, \nu^X)}{(1+\alpha \eta)^2}.$$

    \medskip
    \noindent
    \textbf{(2) Backward step:}
    As explained in~\cite[Section~A.1.4]{CCSW22}, the second step of the Proximal Sampler can be seen as the output of the reverse Gaussian channel.
    Concretely, define the distributions $\mu_0 = \nu^X$ and $\mu_t = \mu_0 \ast \N(0, tI)$ for $0 \le t \le \eta$, so $\mu_\eta = \nu^Y$, and $\mu_t$ evolves following the heat equation: $\partial_t \mu_t = \frac{1}{2} \Delta \mu_t$.    
    Let $\nu_t = \mu_{\eta - t}$, so $\nu_0 = \nu^Y$, $\nu_\eta = \nu^X$, and $\nu_t$ evolves following the backward heat equation~\eqref{Eq:BackwardHeatEq}:
    $$\partial_t \nu_t = -\frac{1}{2} \Delta \nu_t = -\nabla \cdot (\nu_t \nabla \log \nu_t) + \frac{1}{2} \Delta \nu_t.$$
    Now define the reverse Gaussian channel~\eqref{Eq:BackwardBM}:
    $$dX_t = \nabla \log \nu_t(X_t) \, dt + dW_t$$
    so that if $X_t$ evolves along this channel, then its distribution $X_t \sim \rho_t$ evolves following the Fokker-Planck equation:
    $$\partial_t \rho_t = -\nabla \cdot (\rho_t \nabla \log \nu_t) + \frac{1}{2} \Delta \rho_t.$$
    This channel has the property that if we initialize it at a point mass $\rho_0 = \delta_y$ for any $y \in \R^d$, the output at time $t = \eta$ is exactly the conditional distribution $\rho_\eta = \nu^{X \mid Y}(\cdot \mid y)$ that the RGO is supposed to sample from; see the explanation in~\cite[Section~A.1.2]{CCSW22}, see also the exposition in~\cite[Chapter~8.3]{C24}.
    In particular, if we initialize the reverse Gaussian channel from $\rho_0 = \rho_k^Y$, then by construction, the output at time $t = \eta$ is the distribution of the next iterate:
    $$\rho_\eta = \int_{\R^d} \rho_0(y) \, \nu^{X \mid Y}(\cdot \mid y) \, dy = \rho_{k+1}^X.$$    
    Thus, we have interpolated from $\nu_0 = \nu^Y$ to $\nu_\eta = \nu^X$, and from $\rho_0 = \rho_k^Y$ to $\rho_\eta = \rho_{k+1}^X$, as the evolution along the backward heat equation~\eqref{Eq:BackwardHeatEq} or the reverse Gaussian channel~\eqref{Eq:BackwardBM}, which is an instance of the Fokker-Planck channel~\eqref{Eq:FP} with drift function $b_t = \nabla \log \nu_t$ and $c=1$.
    In this case $-c\nabla^2 \log \nu_t + (\nabla b_t)_{\sym} = 0$, so the second term in the time derivative formula~\eqref{Eq:TimeDerivativeFI} vanishes.
    Therefore, along the backward heat equation, using the formula~\eqref{Eq:TimeDerivativeFI} from Lemma~\ref{Lem:TimeDerivativeFI}:
    $$\frac{d}{dt} \FI(\rho_t \,\|\, \nu_t) = -\E_{\rho_t} \left[\left\| \nabla^2 \log \frac{\rho_t}{\nu_t} \right\|^2_{\HS} \right] \le 0.$$
    Integrating from $t=0$ to $t=\eta$ yields:
    \begin{align*}
        \FI(\rho_{k+1}^X \,\|\, \nu^X) = \FI(\rho_\eta \,\|\, \nu_\eta) \,\le\, \FI(\rho_0 \,\|\, \nu_0) = \FI(\rho_k^Y \,\|\, \nu^Y).
    \end{align*}
    Iterating the recursions above yields the desired convergence rate in~\eqref{Eq:FIProximalRateTotal}.    
\end{proof}

%%%%%%%%%%%%%%%%%%%%
\subsection{Proof of Corollary~\ref{Cor:IterationComplexity} (Iteration Complexity of Proximal Sampler in Relative Fisher Information)}
\label{Sec:IterationComplexityProof}

\begin{proof}[Proof of Corollary~\ref{Cor:IterationComplexity}]
Since we assume $\nu^X \propto e^{-g}$ is $\alpha$-SLC and is $L$-log-smooth, we have $\alpha I \preceq \nabla^2 g(x) \preceq LI$ for all $x \in \R^d$.
Since $x^*$ is the minimizer of $g$, we have $\nabla g(x^*) = 0$.
Then from $\rho_0^X = \N(x^*, \frac{1}{L} I)$, we can bound the initial relative Fisher information:
\begin{align*}
    \FI(\rho_0^X \,\|\, \nu) 
    &= \E_{\rho_0^X}\left[\left\| (Lx-\nabla g(x)) - (Lx^*-\nabla g(x^*))\right\|^2\right] \\
    &\le (L-\alpha)^2 \E_{\rho_0^X}\left[\left\| x-x^* \right\|^2\right] \\
    &= \frac{d (L-\alpha)^2}{L} \\
    &\le dL,
\end{align*}
where the first inequality above follows from the fact that $x \mapsto Lx-\nabla g(x)$ is $(L-\alpha)$-Lipschitz, because its Jacobian is $0 \preceq LI - \nabla^2 g(x) \preceq (L-\alpha) I$.

With step size $\eta = \frac{1}{dL}$, we can implement the RGO via rejection sampling with an $O(1)$ expected number of queries, as explained in Section~\ref{Sec:RejectionReview}. 
Then from the exponential convergence rate~\eqref{Eq:FIProximalRateTotal} of the Proximal Sampler from Theorem~\ref{Thm:FIProximalRate}, we have
\begin{align*}
    \FI(\rho_k^X \,\|\, \nu)
    \le \frac{\FI(\rho_0^X \,\|\, \nu^X)}{(1+\alpha \eta)^{2k}}
    \le \frac{dL}{(1+\frac{\alpha}{dL})^{2k}}
    \le dL \exp\left(-\frac{\alpha k}{dL}\right)
\end{align*}
where in the last inequality we use the bound $1+c \ge e^{c/2}$ which holds for $0 \le c = \frac{\alpha}{dL} \le 1$.

Then to reach $\FI(\rho_k^X \,\|\, \nu) \le \varepsilon$, it suffices to set $dL \exp\left(-\frac{\alpha k}{dL}\right) \le \varepsilon$, or equivalently, we run the Proximal Sampler for $k \ge \frac{dL}{\alpha} \log \frac{dL}{\varepsilon}$ iterations, as claimed.
\end{proof}

\paragraph{Acknowledgments.} The author is supported by NSF awards CCF-2403391 and CCF- 2443097.
The author thanks Yihong Wu, Santosh Vempala, Sinho Chewi, and Varun Jog for valuable discussions;
and Tinsel W.\ for guidance.

%%%%%%%%%%%%
\addcontentsline{toc}{section}{References}
% \bibliography{refs}
\bibliography{arxiv-fisher-v2.bbl}

\end{document}